\documentclass[3p, sort&compress]{elsarticle}
\usepackage{amsmath}
\usepackage{amsfonts}
\usepackage{amssymb}
\usepackage{amsthm}
\usepackage[utf8]{inputenc}
\usepackage{cellspace}
\usepackage[small,sc]{caption}
\usepackage{booktabs, multicol, multirow}
\usepackage{enumitem}

\usepackage{graphicx}
\usepackage{epsfig}
\usepackage{algorithm}
\usepackage{algpseudocode}
\usepackage{color}

\usepackage{extpfeil}

\theoremstyle{plain}
\newtheorem{theorem}{Theorem}

\newtheorem{lemma}{Lemma}
\newtheorem{proposition}{Proposition}
\newtheorem{corollary}{Corollary}

\theoremstyle{definition}
\newtheorem{definition}{Definition}

\newtheorem{example}{Example}

\theoremstyle{definition}
\newtheorem{remark}{Remark}
\newtheorem{fact}{Fact}

\begin{document}
\begin{frontmatter}
\title{ \Large{Characterizations of symmetrically partial Boolean functions \\ with exact  quantum query complexity}}

\author{Daowen Qiu\corref{one}}
\author{Shenggen Zheng\corref{two}}

\cortext[one]{{\it E-mail address:} issqdw@mail.sysu.edu.cn (D. Qiu). Corresponding author.}
\cortext[two]{{\it E-mail address:} zhengshg@mail2.sysu.edu.cn (S. Zheng).  Corresponding and also co-first author.}

\address{School of Data and
Computer Science, Sun Yat-sen University, Guangzhou 510006, China}

\begin{abstract}

We give and prove an optimal exact quantum query algorithm with complexity $k+1$ for computing the promise problem (i.e., symmetric and partial Boolean function)  $DJ_n^k$ defined as: $DJ_n^k(x)=1$ for $|x|=n/2$, $DJ_n^k(x)=0$ for $|x|$ in the set $\{0, 1,\ldots, k, n-k, n-k+1,\ldots,n\}$, and it is undefined for the rest cases,  where $n$ is even, $|x|$ is the Hamming weight of $x$.
The case of $k=0$ is the well-known Deutsch-Jozsa problem. We outline all symmetric (and partial) Boolean functions with degrees 1 and 2, and prove their exact quantum query complexity. Then we prove that any symmetrical (and partial) Boolean function $f$ has exact quantum 1-query complexity if and only if $f$ can be computed by the Deutsch-Jozsa algorithm. We also discover the optimal exact quantum 2-query complexity for distinguishing between inputs of Hamming weight $\{ \lfloor n/2\rfloor, \lceil n/2\rceil \}$ and Hamming weight in the set $\{ 0, n\}$ for  all odd $n$. In addition,  a method is provided to determine the degree of any symmetrical (and partial)  Boolean function.

\end{abstract}

\end{frontmatter}

\section{Introduction}

{\em Quantum computing models} can be divided into {\em bounded-error} and {\em exact} versions in terms of their outputs. A bounded-error model means that the mistake probability for any output cannot be beyond an error value given a priori, and an exact model requires its outputs be fully correct always, without any error allowed. {\em Exact quantum computing models} have been studied in the frameworks of {\em quantum finite automata} \cite{AY12,GQZ15} and particularly {\em quantum query models} (for example, \cite{DJ92,BSS03,HS05,BW02,MJM11,AIS13,Amb13}).

The quantum query models are the quantum analog to the classical Boolean decision tree models, so they are also called {\em quantum decision tree} models and are at least as powerful as the classical decision tree models \cite{BW02}. The implementation procedure of a quantum decision tree model is exactly a {\em quantum query algorithm}, and it can be roughly described as: it starts with a fixed starting state $|\psi_s\rangle$ of a Hilbert ${\cal H}$ and will perform the sequence of operations $U_0, O_x, U_1,  \ldots, O_x,U_t$, where $U_i$'s are unitary operators that do not depend on the input $x$ but the query $O_x$  does. This leads to the final state $ |\psi_f\rangle=U_tO_xU_{t-1}\cdots U_1O_xU_0|\psi_s\rangle$. The result is obtained by measuring the final state  $ |\psi_f\rangle$.

A quantum query algorithm ${\cal A}$ {\em  computes exactly} a Boolean function $f$ if its output equals $f(x)$ with probability 1, for all input $x$. ${\cal A}$ {\em computes with bounded-error} $f$ if its output equals $f(x)$ with probability at least $\frac{2}{3}$, for all input $x$. The {\em  exact quantum query complexity} denoted by
$Q_E(f)$ is the minimum number of queries used by any quantum algorithm which
computes $f(x)$ exactly for all input $x$.

For the bounded-error case, quantum query algorithms have been investigated extensively and deeply (for example, \cite{AA15,ABK16,ABA+16,Amb04,Bel11,BDHHMSW01, CD08,DHHM04, DT07, FGG08,Gro96, RS07,  Sho94,Sim94} and the references therein), and some of them have either polynomial speed-up over classical algorithms for computing total Boolean functions. The exact quantum query algorithms for computing total Boolean functions also have been studied \cite{ABK16,ABA+16,Amb13,AGZ14,AIS13,CEMM98,DM06,FGGS98,HKM02,Mid04,MJM11,SVKF15}.   In 2013, as a  breakthrough result, Ambainis \cite{Amb13} has
presented the first example of a Boolean function for which exact quantum algorithms have superlinear advantage over exact classical algorithms, {\em i.e.} $Q_{E}(f)=O(D(f)^{0.8675...})$,
where $D(f)$ denotes the minimum number of queries used by any classical deterministic query algorithm. The result was improved
 by Ambainis {\em et al} \cite{ABK16,ABA+16} to nearly-quadratic separation in 2016.

 Brassard and H{\o}yer \cite{BH97} gave an example of a partial function whose exact quantum query complexity is exponentially lower than its classical randomized query complexity.
However, for computing partial Boolean functions, there can be more than exponential separation between exact quantum and classical deterministic query complexity, and the first result was the well-known Deutsch-Jozsa algorithm \cite{DJ92}.

 Deutsch-Jozsa problem \cite{DJ92} can be  described as a partial Boolean function $DJ_n^0: \{0,1\}^n\rightarrow \{0,1\}$ defined as: $n$ is even, and $DJ_n^0(x)=1$ for $|x|=\frac{n}{2}$ and $DJ_n^0(x)=0$ for $|x|=0~\text{or} ~n$, and the other cases are undefined,  where $|x|$ is the Hamming weight of $x$. Deutsch-Jozsa problem has attracted a lot of research and discussion (for example, \cite{CKH98,MJM11,Ber14}), and the physical realization was implemented in \cite{PMG+15}. Montanaro, Jozsa, and Mitchison  \cite{MJM11} generalized the Deutsch-Jozsa problem to another partial Boolean function, say $DJ_n^1: \{0,1\}^n\rightarrow \{0,1\}$ defined as $DJ_n^0$ except for $DJ_n^1(x)=0$ for $|x|=0,1,n-1,n$.  Also, Montanaro {\em et al} \cite{MJM11} designed an exact quantum 2-query algorithm to compute it by using an analytical method.

Aaronson and Ambainis \cite{AA14} have showed that there can be at most a quadratic separation between quantum and classical bounded-error algorithms for computing any symmetric partial boolean function.
In this paper we will study symmetric partial boolean function for the exact computing cases.
Ambainis \cite{Amb98} showed that almost all total Boolean functions have high approximate degree, so, we are also interested in partial Boolean functions with lower degree. Indeed, partial Boolean functions have also been called as promise problems \cite{ESY84,Go06}, and both symmetric Boolean functions and partial Boolean functions have had important applications in cryptography (for example, \cite {CV05,ESY84,Go06}).

\subsection{Definitions}

Let $f$ be a Boolean function from $D\subseteq\{0,1\}^n$ to $\{0, 1\}$.  If $D=\{0,1\}^n$, then $f$ is called a total Boolean function. Otherwise, $f$ is called a partial Boolean function or a promise problem \cite{ESY84,Go06} and  $D$ is referred to as the domain of definition or promised set.

A (partial) Boolean function $f$ is called symmetric if $f(x)$  only depends  on the Hamming weight of $x$, i.e., $|x|$. Some characteristics of the symmetric Boolean functions were given in, for example,  \cite {CV05}. Some common symmetric functions over $\{0,1\}^n$ are listed as follows.

\begin{itemize}
  \item $OR_n(x)=1$ if and only if $|x|\geq 1$;
  \item $AND_n(x)=1$ if and only if $|x|= n$;
 \item $PARITY_n(x)=1$ if and only if $|x|$ is odd;
 \item $MAJ_n(x)=1$ if and only if $|x|>n/2$;
\item $EXACT_n^k(x)=1$ if and only if $|x|=k$, where $0\leq k\leq n$;
\item $THRESHOLD_n^k(x)=1$ if and only if $|x|\geq k$, where $0\leq k\leq n$.
\end{itemize}

\begin{remark}
In \cite{BWY15}, {\em partially symmetric Boolean functions} were studied and the definition is: For a subset $J\subseteq [n]:=
\{1, \ldots , n\}$, a function $f: \{0, 1\}^n \rightarrow \{0, 1\}$ is $J$-symmetric if permuting the labels
of the variables of $J$ does not change the function. So, a partially symmetric Boolean function is a total function but its symmetric property is partial. If $J=[n]$, then it is exactly a symmetric Boolean function.

\end{remark}

So, different from partially symmetric Boolean functions \cite{BWY15}, the functions $\text{DJ}_n^0$ and $\text{DJ}_n^1$ above are both symmetric and  partial,  called  {\em symmetrically  partial  Boolean functions} (i.e. promise problems) in this paper and the exact definition can be described as follows.

\begin{definition}

Let $f:\{0, 1\}^n \rightarrow \{0, 1\}$ be a partial Boolean function, and let $D\subseteq \{0, 1\}^n$ be its domain of definition. If for any $x\in D$ and for any $y\in \{0, 1\}^n$ with $|x|=|y|$, it holds that $y\in D$ and $f(x)=f(y)$, then $f$ is called a {\em  symmetrically  partial  Boolean function}. When $D=\{0,1\}^n$,  $f$ is  an symmetric function.

\end{definition}

So, a  symmetrically  partial Boolean function equals a symmetric and partial Boolean function, and has been called a {\em promise problem} \cite{ESY84,Go06}.
Clearly, if $f:\{0, 1\}^n \rightarrow \{0, 1\}$ is a symmetrically  partial function, then its domain of definition has the version $\{x: |x|=k_1,k_2,\ldots,k_l\}$ for some $0\leq k_i \leq n$ with $i=1,2,\ldots,l$.

{\em Isomorphism} is useful in the study of query complexity, and two partial  functions $f$ and $g$ over $\{0, 1\}^n$  are {\em isomorphic}  if they are equal up to negations and permutations of
the input variables, and negation of the output variable.  % equivalently speaking,  if: (1) $f(x)=g(\bar{x})$ for any $x$ in its domain of definition, where $\bar{x}=$ is the negation of $x$%

\begin{fact} \label{fact1}
For any two partial  functions $f,g$ over $\{0,1\}^n$, if they are isomorphic, then they have the same (exact) quantum query complexity.
\end{fact}
\begin{proof}
Let  $g(x)=(\neg)f(\pi((\neg) x_1, (\neg) x_2,\ldots, (\neg) x_n ))$ where $\pi$ is a permutation. Suppose that there is  a $t$-queries quantum algorithm ${\cal A}$ that computes $f(x)$, and let ${\cal A}(x)$ represent the output (0 or 1) for input $x$.
Now for any $x$ in the domain of definition of $g$, we consider the following $t$-queries quantum algorithm ${\cal A}'$:
\begin{equation}
    {\cal A}'(x)=(\neg) {\cal A} U_1U_0(x),
\end{equation}
where $U_0(x)=((\neg) x_1, (\neg) x_2,\ldots, (\neg) x_n )$ and $U_1(x)=\pi(x)$. It is clear that ${\cal A}'$ computes exactly function $g$.
\end{proof}

\begin{remark}
Given a partial symmetric function $f:\{0,1\}^n\rightarrow \{0,1\}$, with the domain $D$ of definition,  it can be equivalently described by a vector $(b_0,b_1,\ldots,b_n)\in\{0,1,*\}^{n+1}$, where $f(x)=b_{|x|}$,  i.e. $b_k$ is the value of $f(x)$ when $|x|=k$, and $f(x)$ is `undefined' for $b_{|x|}=*$. In the interest of simplicity, sometimes we will use the vector to denote a symmetrically  partial function in this article.

\end{remark}

Concerning the $n$-bit symmetrically  partial  functions, it is clear that the following functions are isomorphic to each other:

\begin{itemize}
  \item $(b_0,b_1,\ldots,b_n)$;
  \item $(b_n,b_{n-1},\ldots,b_0)$;
 \item $(\bar{b}_0,\bar{b}_1,\ldots,\bar{b}_n)$;
 \item $(\bar{b}_n,\bar{b}_{n-1},\ldots,\bar{b}_0)$.
\end{itemize}

We need to introduce some complexity measures for symmetrically  partial  functions.

\begin{definition}

Let $f$ be a partial function with a domain of definition $D\subseteq\{0,1\}^n$.
For $0\leq \varepsilon < 1/2$, we say a real multilinear polynomial $p$ approximates $f$ with error $\varepsilon$ if:
\begin{enumerate}
\item[(1)]  $|p(x)-f(x)|\leq \varepsilon$ for all $x\in D$;
\item[(2)]  $0\leq p(x)\leq 1$ for all $x\in\{0,1\}^n$.
\end{enumerate}
The approximate degree of $f$ with error $\varepsilon$, denoted by $\widetilde{\text{deg}}_{\varepsilon}(f)$, is the minimum degree among all real multilinear polynomials  that approximate $f$ with error $\varepsilon$.

\end{definition}

%A real multilinear polynomial  $p:\mathbf{R}^n\to \mathbf{R}$ represents $f$ if (1) $p(x)=f(x)$ for all $x\in D$, and (2) $0\leq p(x)\leq 1$ for all $x\in\{0,1\}^n$.    The {\em exact degree} of $f$, denoted $\overline{\text{deg}}(f)$, is the minimum degree of all real multilinear polynomials  representing $f$.

Clearly, if $\varepsilon=0$, then $\widetilde{\text{deg}}_{0}(f)$ is the {\em exact degree} of $f$. Furthermore,
if $D=\{0,1\}^n$, i.e. $f$ is a total function, then the exact degree of $f$ is exactly the degree of $f$ as usual \cite{BW02}, denoted by $\text{deg}(f)$. In the interest of simplicity, sometimes we just identity $\widetilde{\text{deg}}_{0}(f)$ with $\text{deg}(f)$ for any partial Boolean function $f$, since  no confusion leads.

\subsection{Preliminaries}

Let input $x=x_1\cdots x_n\in\{0,1\}^n$ for some fixed $n$. We will consider a Hilbert space ${\cal H}$ with basis states $|i,j\rangle$ for $i\in\{0,1,\ldots,n\}$ and $j\in\{1,\cdots, m\}$ (where $m$ can be chosen arbitrarily). A query $O_x$ to an input $x\in\{0,1\}^n$ will be formulated as the following unitary transformation:
\begin{itemize}
  \item $O_x|0,j\rangle=|0,j\rangle$;
  \item $O_x|i,j\rangle=(-1)^{x_i}|i,j\rangle$ for $i\in\{1,2,\cdots, n\}$.
\end{itemize}

A  quantum query algorithm ${\cal A}$ which uses $t$ queries for an input $x$ consists of a sequence of
 unitary operators $U_0, O_x, U_1,  \ldots, O_x,U_t$, where $U_i$'s  do not depend on the input
$x$ and the query $O_x$   does. The algorithm will start in a fixed starting state $|\psi_s\rangle$ of ${\cal H}$ and will perform the above sequence of operations. This leads to the final state
\begin{equation}
   |\psi_f\rangle=U_tO_xU_{t-1}\cdots U_1O_xU_0|\psi_s\rangle.
\end{equation}
The final state is then measured with
a measurement $\{M_0, M_1\}$.  For an  input $x\in\{0,1\}^n$, we denote ${\cal A}(x)$  the output of the  quantum query algorithm  ${\cal A}$.  Obviously,
$Pr[{\cal A}(x)=0] =\|M_0|\psi_f\rangle\|^2$ and $Pr[{\cal A}(x)=1] =\|M_1|\psi_f\rangle\|^2=1-Pr[{\cal A}(x)=0]$.
We say that the quantum query algorithm ${\cal A}$ computes $f$ within an
error $\varepsilon$ if for every input $x\in\{0,1\}^n$ it holds that $Pr[{\cal A}(x)=f(x)]\geq 1-\varepsilon$. If $\varepsilon=0$, we says that the  quantum algorithm is exact.
%(By the principle of delayed measurement, the effective of measurement being performed in the middle of computation sometimes is equivalent to that of measuring it  at the end of computation \cite{NC00}.)%
For more details on quantum query complexity, we may refer to \cite{Amb13,BW02,BSS03,MJM11}.

{\em Quantum query models} are one of most important computing models in quantum computing. In this complexity models \cite{BW02}, an algorithm is charged for ``queries" to the input bits, while any intermediate computation is considered as free. For many functions one can obtain large quantum speed-ups in the case algorithms are allowed a constant small probability of error (bounded error). As the most famous example, Grover's algorithm \cite{Gro96} computes the $n$-bit $\mbox{OR}$ function with $O(\sqrt {n})$ queries in the bounded-error mode, while any classical algorithm needs $\Omega(n)$ queries. The model of  {\em exact quantum query}, where the algorithms must output the correct answer with certainty for every possible input, seems to be more intriguing \cite{BH97,CEMM98,DJ92}. It is much more difficult to come up with  exact quantum algorithms that outperform  classical deterministic algorithms.

In the exact quantum query complexity, it was recognized that the best quantum
speed-up for computing total functions was by a factor of 2 for many years \cite{FGGS98}. In a breakthrough result, Ambainis  has
presented the first example of a Boolean function $f:\{0,1\}^n\to \{0,1\}$ for which exact quantum algorithms have superlinear advantage over classical deterministic algorithms \cite{Amb13}. The result was improved in 2016 \cite{ABK16,ABA+16}. Based on the results in \cite{MJM11,AIS13}, Ambainis, Gruska, and Zheng \cite{AGZ14} have verified that exact quantum algorithms have certain advantage for most of Boolean functions.

Ambainis  \emph{et al}~\cite{AIS13}
 have developed optimal exact quantum algorithms for computing functions
 $EXACT_n^k$ and $THRESHOLD_n^k$,   which are  to determine
whether an $n$-bit string has Hamming weight exactly $k$ and to determine
whether an $n$-bit string has Hamming weight at least $k$. The complexity is:

\begin{itemize}
\item $Q_E(EXACT_n^k)=\max(k,n-k)$;
\item $Q_E(THRESHOLD_n^k)=\max(k,n-k+1)$.
\end{itemize}

If $f$ is allowed to be a partial function,  the Deutsch-Jozsa  algorithm \cite{DJ92} proved that there can be more than exponential
separation between exact quantum and  classical deterministic  query complexity. Some generalizations \cite{MJM11,CKL00,ZQ14,GQZ15} of the  Deutsch-Jozsa  problem were also investigated, and we will indicate them carefully if there exist relations to our results.

\subsection{Our main results and proof methods}

A general generalization of Deutsch-Jozsa problem is the following partial symmetric  function:
\begin{equation} \label{DJF}
\text{DJ}_n^k(x)=\left\{\begin{array}{ll}
                    1 &\ \text{if}\ |x|=n/2, \\
                    0 &\ \text{if}\ |x|\leq k\ \text{or}\  |x|\geq n-k,
                  \end{array}
 \right.
 \end{equation}
 where  $n$ is even and $0\leq k<n/2$.

Clearly, when $k=0$, it is the Deutsch-Jozsa problem, and when $k=1$, it equals the problem given by  Montanaro \emph{et al}~\cite{MJM11}.

Our first main result is as follows.

\begin{theorem}\label{QE(DJ)}
The exact quantum query complexity of $\text{\em DJ}_n^k$ satisfies:
\begin{equation}
    Q_E(\text{\em DJ}_n^k)=k+1.
\end{equation}

However, the classical deterministic query complexity for $\text{\em DJ}_n^k$ is:
\begin{equation}
    D(\text{\em DJ}_n^k)=n/2+k+1.
\end{equation}
\end{theorem}

\begin{remark}
When $k=1$, Montanaro \emph{et al.}~ \cite{MJM11} designed an exact quantum 2-query algorithm to compute it, and their method is somewhat complicated for deriving a unitary operator from solving a system of equations, but the optimality with 2-query was not verified. Our result also shows the algorithm by Montanaro {\em et al} \cite{MJM11} is optimal.

\end{remark}

{\em Proof method of Theorem 1:}  Using the exact quantum query algorithms for computing $EXACT_n^k$ and $THRESHOLD_n^k$ due to Ambainis  \emph{et al}~\cite{AIS13}, we can give an exact quantum $(k+1)$-query algorithm for computing $\text{\em DJ}_n^k$.
On the other hand,  we will prove that $\widetilde{\text{deg}}_{0}(\text{DJ}_n^k)\geq 2k+2$, and therefore
 $Q_E(\text{DJ}_n^k)\geq \frac{\widetilde{\text{deg}}_{0}(\text{DJ}_n^k)}{2}=k+1$.

A natural question is what common characters are for the Boolean functions with the same exact quantum query complexity?  Due to the importance and simplicity of symmetric functions,  here we consider the case of exact quantum 1-query complexity for all symmetrical and partial functions.

Therefore, the question is what can be solved with exact quantum 1-query complexity? Notably,  Aaronson, Ambainis, Iraids, and Kokainis \cite{AAIK15} recently proved that a partial Boolean function $f$ is computable by a 1-query quantum
algorithm with error bounded by $\varepsilon<1/2$ if and only if $f$ can be approximated by a {degree-2} polynomial
with error bounded by $\varepsilon'<1/2$.

We can pose the question more precisely: if an exact quantum 1-query algorithm $\cal{A}$ computes a symmetrically  partial function $f$, then, can any symmetrical and partial function $g$ with $Q_E(g)=1$ be computed by $\cal{A}$? Our second main result  answers this question  as follows.

\begin{theorem}
Any symmetric and partial Boolean function $f$ has $Q_E(f)=1$ if and only if $f$ can be computed by the Deutsch-Jozsa algorithm.
\end{theorem}

To prove the above theorem, we prove the following three results.

\begin{theorem}
Let $n>1$  and let $f:\{0,1\}^n \to\{0,1\}$  be an $n$-bit symmetric and partial Boolean function. Then:

  (1) $\text{\em deg}(f)=1$ if and only if $f$ is isomorphic to the function $f^{(1)}_{n,n}$;

  (2) $\text{\em deg}(f)=2$ if and only if $f$ is isomorphic to one of the functions
\begin{align}
&f^{(1)}_{n,k}(x)=\left\{\begin{array}{ll}
                    0 &\text{\em if}\ |x|=0, \\
                    1 &\text{\em if}\ |x|= k,
                  \end{array}
 \right.\\
&f^{(2)}_{n,k}(x)=\left\{\begin{array}{ll}
                    0 &\text{\em if}\ |x|=0, \\
                    1 &\text{\em if}\ |x|= k\ \text{\em or } |x|=k+1,
                  \end{array}
 \right.\\
 &f^{(3)}_{n,l}(x)=\left\{\begin{array}{ll}
                    0 &\text{\em if}\ |x|=0 \ \text{\em or } |x|= n,\\
                    1 &\text{\em if}\ |x|= l,
                  \end{array}
 \right.\\
&f^{(4)}_{n}(x)=\left\{\begin{array}{ll}
                    0 &\text{\em if}\ |x|=0 \ \text{\em or } |x|= n,\\
                    1 &\text{\em if}\ |x|=\lfloor n/2\rfloor \ \text{\em or } |x|=\lceil n/2\rceil,
                  \end{array}
 \right.
 \end{align}
where $n-1\geq k\geq \lfloor n/2\rfloor$, and $\lceil n/2\rceil\geq l\geq \lfloor n/2\rfloor$.

\end{theorem}

 With the above theorem we can further prove the two theorems as follows.

\begin{theorem}
Let $n$ be even  and let $f:\{0,1\}^n \to\{0,1\}$  be an $n$-bit symmetric and partial function. Then $Q_E(f)=1$ if and only if $f$ is isomorphic to one of these functions: $f^{(1)}_{n,k}$ and $f^{(3)}_{n,n/2}$,
 where $ k\geq n/2$.
\end{theorem}

\begin{theorem}
Let $n$ be odd and let $f:\{0,1\}^n \to\{0,1\}$  be an $n$-bit symmetric and partial  function. Then $Q_E(f)=1$ if and only if $f$ is isomorphic to one of the functions $f^{(1)}_{n,k}$,
 where $ k\geq \lceil n/2 \rceil$.
\end{theorem}

\vskip 5mm

{\em Proof method of Theorem 2:}
Suppose that  $Q_E(f)=1$. Then:

(1) for $n$ being odd, $f$ is isomorphic to one of the functions $f^{(1)}_{n,k}$, where $ k\geq \lceil n/2 \rceil$;

(2) for $n$ being even, $f$ is isomorphic to one of these functions: $f^{(1)}_{n,k}$ and $f^{(3)}_{n,n/2}$,
 where $ k\geq n/2$.

If $f$ is isomorphic to $f^{(1)}_{n,k}$, then we pad $2k-n$ zeros to the input of the function $f^{(1)}_{n,k}$. As a result, it is equivalently to compute the function $f^{(1)}_{2k,k}$.  Clearly $f^{(1)}_{2k,k}$ is a more special problem than Deutsch-Jozsa  problem, and therefore it can also be computed by the Deutsch-Jozsa algorithm.

If $f$ is isomorphic to $f^{(3)}_{n,n/2}$ ($n$ being even), then it is just the Deutsch-Jozsa  problem.

Consequently, $Q_E(f)=1$ implies that $f$ can always be computed by the Deutsch-Jozsa algorithm.

On the other hand, if $f$ can be computed by the Deutsch-Jozsa algorithm, then $Q_E(f)=1$ (here we omit the ordinary case of $f$ being a constant function, and therefore we always suppose $Q_E(f)> 0$).

\subsection{Problems}

A problem for further study is to characterize the symmetric and partial Boolean functions by exact quantum $(k+1)$-query complexity for $k\geq 0$. It can be described  more precisely in the following.

Let $f:\{0,1\}^n\to \{0,1\}$  be an $n$-bit partial symmetric  Boolean function with domain of definition $D$, and let $0\leq k< \lfloor n/2\rfloor$. Then, for  $2k+1\leq \text{ deg}(f)\leq 2(k+1)$, how to characterize $f$ by giving all functions with  degrees from $2k+1$ to $2k+2$? A possible conjecture by means of generalization is  that $f$ may be isomorphic to one of the following functions:
\begin{align}
&f^{(1),k}_{n,m}(x)=\left\{\begin{array}{ll}
                    0 &\text{if}\ |x|\leq k, \\
                    1 &\text{if}\ |x|= m,
                  \end{array}
 \right.\\
&f^{(2),k}_{n,m}(x)=\left\{\begin{array}{ll}
                    0 &\text{if}\ |x|\leq k, \\
                    1 &\text{if}\ |x|= m\ \text{\em or } |x|=m+1,
                  \end{array}
 \right.\\
&f^{(3),k}_{n,l}(x)=\left\{\begin{array}{ll}
                    0 &\text{if}\ |x|\leq k \ \text{\em or } |x|\geq n-k,\\
                    1 &\text{if}\ |x|=l,
                  \end{array}
 \right.\\
&f^{(4),k}_{n}(x)=\left\{\begin{array}{ll}
                    0 &\text{if}\ |x|\leq k \ \text{\em or } |x|\geq n-k,\\
                    1 &\text{if}\ |x|=\lfloor n/2\rfloor \ \text{\em or } |x|=\lceil n/2\rceil,
                  \end{array}
 \right.
 \end{align}
where $m\geq \lfloor n/2\rfloor$ and $\lfloor n/2\rfloor\leq l\leq \lceil n/2\rceil$.

Moreover, how to characterize the symmetrically  partial Boolean function $f$ in terms of its exact quantum $(k+1)$-query complexity? Can any symmetrically  partial Boolean function $f$ with exact quantum $(k+1)$-query complexity  be computed by the presented exact quantum $(k+1)$-query algorithm for computing the generalized Deutsch-Jozsa problem $DJ_n^k$?

Related are another two questions:
\begin{enumerate}
   \item $2\leq Q_E(f^{(2)}_{n,k})\leq 4$ for $n/4\leq k<n$ and $2\leq Q_E(f_n^{(4)})\leq 5$ will be verified in the article. How to determine $Q_E(f^{(2)}_{n,k})$ and $Q_E(f_n^{(4)})$?

   \item  For the function $\text{DW}_n^{k,l}$, defined as:
$$
\text{DW}_n^{k,l}(x)=\left\{\begin{array}{ll}
                    0 &\text{if}\ |x|=k, \\
                    1 &\text{if}\ |x|=l,
                  \end{array}
 \right.
$$
    we will give some optimal exact quantum query algorithms for some special choices of $k$ and $l$. Can we give optimal exact quantum query algorithms for any $k$ and $l$?

 \end{enumerate}

\subsection{Organization}

%The remainder of this paper is organized as follows. In Section~2 we give a method for finding the approximate degree of  multilinear polynomial  of partial symmetric  Boolean  functions. In Section~3 we investigate functions $\text{DJ}_n^k$ and $\text{DJ2}_n^k$. In Section~4 we study partial symmetric  Boolean functions with exact degree less equal to 2 and  functions that can be computed by  1-query exact quantum algorithm. We consider the functions $\text{DW}_n^{k,l}$ in  Section~5. Finally, Section~6 contains a conclusion and a number of problems related for further study.

The remainder of this article is organized as follows. In Section~2 we study the representation of  symmetrically  partial  Boolean  functions with multilinear polynomials, and give a method for finding the approximate (and exact) degree of symmetrically  partial Boolean  functions. Then in Section~3 we investigate the exact quantum and classical deterministic query complexity of a generalized Deutsch-Jozsa problem, that is, the function $\text{DJ}_n^k$, and we present an optimal exact quantum $(k+1)$-query algorithm to compute $\text{DJ}_n^k$, but its classical deterministic query complexity is $n/2+k+1$. After that, in Section~4 we give all symmetrically  partial Boolean functions with exact degree 1 or 2 in the sense of isomorphism. By combining the results of Section~4, in Section 5 we study the exact quantum query complexity for symmetrically  partial Boolean functions with exact degree 1 or 2, and in particular, we present all symmetrically  partial  Boolean functions with exact quantum 1-query complexity as well as prove that these function can be computed by the Deutsch-Jozsa algorithm. In addition, we in Section~6 study further the exact quantum query complexity for some symmetrically  partial Boolean  functions.

\section{Degree of  polynomials for symmetric and partial  functions}

First we study the exact degree of  symmetrically  partial functions.  We can use the method of symmetrization \cite{MP68} to prove the following lemma.

\begin{lemma}\label{Lemma-degree}
For any symmetrically  partial  function $f$ over $\{0,1\}^n$ with domain of definition $D$, suppose $\widetilde{\text{deg}}_{0}(f)\leq d$. Then  there exists a real multilinear polynomial $q$ representing $f$ and $q$ can be written as
\begin{equation}
    q(x)=c_0+c_1V_1+c_2V_2+\cdots+c_dV_d,
\end{equation}
where $c_i\in \mathbf{R}$, $V_{i}=\Sigma_{j_1j_2...j_i\in \{1,2,...,n\}^{i}}x_{j_1}x_{j_2}...x_{j_i}$ where any $j_1j_2...j_i\in \{1,2,...,n\}^{i}$ is without repeated number, $1\leq i\leq d$, for example,
$V_1=x_1+\cdots+x_n$, $V_2=x_1x_2+x_1x_3+\cdots+ x_{n-1}x_n$, etc.
\end{lemma}
\begin{proof}
Let $p$ be a multilinear polynomial representing $f$ and let $\text{deg}(p)=\widetilde{\text{deg}}_{0}(f)=d$.  If $\pi$ is some permutation and $x=(x_1,\ldots, x_n)$, then $\pi(x)=(x_{\pi(1)},\ldots, x_{\pi(n)})$.  Let $S_n$ be the set of all $n!$ permutations. For any $x\in\{0,1\}^n$, the symmetrization of $p$ is,
 \begin{equation}
    p^{\text{sym}}(x)=\frac{\sum_{\pi\in S_n} p(\pi(x))}{n!}.
\end{equation}
Clearly, $0\leq p(x)\leq 1$ implies $0\leq p^{\text{sym}}(x)\leq 1$ for $x\in\{0,1\}^n$.
Since $f$ is symmetric,  $x\in D$ implies $\pi(x)\in D$. For all $x\in D$, we have $f(\pi(x))=f(x)$.  Since $p$ represents $f$, for any $x\in D$, we have $p(\pi(x))=f(\pi(x))=f(x)=p(x)$. Therefore, for any $x\in D$,
 \begin{equation}
    p^{\text{sym}}(x)=\frac{\sum_{\pi\in S_n} p(\pi(x))}{n!}=\frac{\sum_{\pi\in S_n} p(x)}{n!}=p(x)=f(x).
\end{equation}
So $p^{\text{sym}}$ can represent  $f$. Let the multilinear polynomial $q=p^{\text{sym}}$. According to Minsky and Papert's result \cite{MP68} (also Lemma 2 in \cite{BW02}), $q$ can be written as
\begin{equation}
    q(x)=c_0+c_1V_1+c_2V_2+\cdots+c_dV_d.
\end{equation}
Therefore, the lemma has been proved.
\end{proof}

\begin{example}
Let us  give an example to find out $\widetilde{\text{deg}}_{0}(f)$  for $f=\text{DJ}_n^{0}$, which is the   Deutsch-Jozsa problem. We prove that $\widetilde{\text{deg}}_0(\text{DJ}_n^{0})\leq 2$.
Therefore, we assume that there is a multilinear polynomial $q(x)=c_0+c_1V_1+c_2V_2$ representing $\text{DJ}_n^{0}$.
For $|x|=0$, we have $q(x)=c_0=f(x)=0$.
For $|x|=n$, we have $q(x)={n\choose 0}c_0+{n\choose 1}c_1+{n\choose 2}c_2={n\choose 1}c_1+{n\choose 2}c_2=0$.
For $|x|=\frac{n}{2}$, we have $q(x)={n/2\choose 0}c_0+{n/2\choose 1}c_1+{n/2\choose 2}c_2={n/2\choose 1}c_1+{n/2\choose 2}c_2=1$.
Therefore, we need to find out the solution of the following linear system of equations:
 \begin{equation}
\left\{\begin{array}{ll}
                    c_0=0,\\
                    {n\choose 0}c_0+{n\choose 1}c_1+{n\choose 2}c_2=0,\\
                    {n/2\choose 0}c_0+{n/2\choose 1}c_1+{n/2\choose 2}c_2=1.
                  \end{array}
 \right.
 \end{equation}
It is easy to obtain that $c_0=0$, $c_1=\frac{4(n-1)}{n^2}$, $c_2=-\frac{8}{n^2}$ and
$
 q(x)=\frac{4(n-1)}{n^2}V_1-\frac{8}{n^2}V_2
$
 representing $\text{DJ}_n^{0}$.  Therefore, $\widetilde{\text{deg}}_0(\text{DJ}_n^{0})\leq 2$.

 Suppose that $\widetilde{\text{deg}}_0(\text{DJ}_n^{0})\leq 1$. Then there exists a multilinear polynomial $q(x)=c_0+c_1V_1$ representing $\text{DJ}_n^{0}$. We need to get the solution for the following linear group of equations:
  \begin{equation}
\left\{\begin{array}{ll}
                    c_0=0,\\
                    {n\choose 0}c_0+{n\choose 1}c_1=0,\\
                    {n/2\choose 0}c_0+{n/2\choose 1}c_1=1.
                  \end{array}
 \right.
 \end{equation}
It is easy to deduce that there is no solution. Therefore, $\widetilde{\text{deg}}_0(\text{DJ}_n^{0})> 1$, and  consequently $\widetilde{\text{deg}}_0(\text{DJ}_n^{0})=2$.  The example ends.
 \end{example}

For any total function $f$, with the next lemma it has been proved \cite{BBC+98} (or see \cite{BW02}) that  $Q_E(f)\geq \frac{1}{2}\text{deg}(f)$.

\begin{lemma}{\em \cite{BBC+98,BW02}}\label{lm-bw-1}
Let ${\cal A}$ be a quantum query algorithm that makes $t$ queries. Then there exist complex-valued $n$-variate multilinear polynomials $\alpha_i$ of degree at most $t$, such that the final state of ${\cal A}$  is
\begin{equation}
    \sum_{i\in \{0,1\}^m} \alpha_i(x)|i\rangle
\end{equation}
for every input $x\in\{0,1\}^n$.
\end{lemma}

Indeed, according to the proof of Theorem 17 in \cite{BW02} (also refer to \cite{BBC+98}), the following result still holds, and we also give a similar proof.

\begin{lemma}\label{Lemma-lower-bound2}
For any partial Boolean function  $f$,  $Q_{\varepsilon}(f)\geq \frac{1}{2}\widetilde{\text{deg}}_{\varepsilon}(f)$, where $Q_{\varepsilon}(f)$ denotes the quantum query complexity for $f$ with bounded-error $\varepsilon$.
\end{lemma}

\begin{proof}
Consider a  $Q_{\varepsilon}(f)$-query quantum algorithm for $f$ with error $\varepsilon$.  Let $S$ be the set of basis states corresponding to a 1-output. Consider the polynomial $p(x)=\sum_{i\in S}|\alpha_i(x)|^2$, which is the probability that the algorithm outputs 1. If $x\in D$ and $f(x)=1$, then $p(x)\geq 1-\varepsilon$.  If $x\in D$ and $f(x)=0$, then $p(x) \leq \varepsilon$.
Therefore, $|p(x)-f(x)|\leq \varepsilon$ for all $x\in D$. Since the algorithm procedure to get the last state for any input $x$ is the implementation of a sequence of unitary operators, it is clear that $0\leq p(x) \leq 1$ for all $x\in\{0,1\}^n$.  So polynomial $p(x)$ approximates $f$ with error  $\varepsilon$.  According to Lemma \ref{lm-bw-1}, the $\alpha_i$ are polynomials of degree no more than $Q_{\varepsilon}(f)$, therefore $p(x)$ is a polynomial of degree no more than  $2Q_{\varepsilon}(f)$. Consequently, we have
\begin{equation}
    \widetilde{\text{deg}}_{\varepsilon}(f)\leq {\text{deg}}(p)\leq 2Q_{\varepsilon}(f),
\end{equation}
and the lemma has been proved.
\end{proof}

In particular, when $\varepsilon=0$ we have the following special case.

\begin{lemma}\label{Lemma-lower-bound1}
For any symmetrically  partial function  $f$,  $Q_E(f)\geq \frac{1}{2} \widetilde{\text{deg}}_{0}(f)$.
\end{lemma}

We have proved that $\widetilde{\text{deg}}_{0}(\text{DJ}_n^{0})=2$. According to the above lemma, $Q_E(\text{DJ}_n^{0})\geq \frac{1}{2} \widetilde{\text{deg}}_{0}(\text{DJ}_n^{0})=1$. It is known  that $Q_E(\text{DJ}_n^{0})\leq 1$ \cite{DJ92}. Therefore, we can use the above lemma to conclude $Q_E(\text{DJ}_n^{0})=1$.

Now we deal with the case of approximating representation.

\begin{lemma}\label{Lemma-degree-II}
For any symmetrically  partial  Boolean  function $f$ over $\{0,1\}^n$ with domain of definition $D$,
 suppose $\widetilde{\text{deg}}_{\varepsilon}(f)= d$. Then  there exists a real multilinear polynomial $q$ approximates $f$ with error $\varepsilon$ and $q$ can be written as
\begin{equation}
    q(x)=c_0+c_1V_1+c_2V_2+\cdots+c_dV_d
\end{equation}
where $c_i\in \mathbb{R}$, $V_1=x_1+\cdots+x_n$, $V_2=x_1x_2+x_1x_3+\cdots+ x_{n-1}x_n$, etc.
\end{lemma}
\begin{proof}
The proof is similar to that of Lemma \ref{Lemma-degree}. For the readability, we outline it again.
Let $p$ be a multilinear polynomial with degree $d$ that approximates $f$ with error $\varepsilon$.
The  symmetrization of $p$ is
\begin{equation}
    p^{\text{sym}}(x)=\frac{\sum_{\pi\in S_n} p(\pi(x))}{n!}.
\end{equation}
If $x\in D$, then $|p(x)-f(x)|\leq \varepsilon$. Since $f$ is symmetric, we have $|p^{\text{sym}}(x)-f(x)|=|p(x)-f(x)|\leq \varepsilon$.
Since $0\leq  p(\pi(x)) \leq 1$ for all $x\in\{0,1\}^n$, we have $0\leq p^{\text{sym}}(x)\leq 1$ for all $x\in\{0,1\}^n$.
 According to Minsky and Papert's result \cite{MP68} (also Lemma 2 in \cite{BW02}), $p^{\text{sym}}$ can be written as
\begin{equation}
   p^{\text{sym}}(x)=c_0+c_1V_1+c_2V_2+\cdots+c_dV_d.
\end{equation}
Therefore, $p^{\text{sym}}$ is the polynomial required.
\end{proof}

It is important to determine the approximate degree of symmetrically  partial functions. The following lemma shows if or not a symmetrically  partial function has degree $d$.

\begin{lemma}\label{LP}
For any symmetrically  partial function $f$ over $\{0,1\}^n$ with domain of definition $D$, and for the fixed $d$ and $0\leq \varepsilon<1/2$, there is a linear programming algorithm to discover whether or not there exists
\begin{equation}
    q(x)=c_0+c_1V_1+c_2V_2+\cdots+c_dV_d=\sum_{k=0}^d c_k V_k
\end{equation}
 approximating $f$ with error  $\varepsilon$.
\end{lemma}
\begin{proof}
Suppose that $f$ is fully described by the vector $(b_0,b_1,\ldots,b_n)\in\{0,1,*\}^{n+1}$, where  $f(x)=b_i$ for $|x|=i$.  For input $x$, $V_k={|x| \choose k}$. If there exists a polynomial $q$ with degree $d$  approximating $f$ with error  $\varepsilon$,   then for $0\leq i\leq n$,  $q(x)$ satisfies the following  inequalities and equalities:
\begin{enumerate}
  \item $0\leq q(x)= \sum_{k=0}^d c_k {i\choose k}\leq \varepsilon$ if $b_i=0$;
  \item $1-\varepsilon\leq q(x)= \sum_{k=0}^d c_k {i\choose k}\leq 1$ if $b_i=1$;
  \item  $0\leq q(x)= \sum_{k=0}^d c_k {i\choose k}\leq 1$ if $b_i=*$.
\end{enumerate}
 Therefore, it suffices to verify whether the polyhedra has solution or not. It is easy to transfer the above polyhedra to the normal form, i.e., $P=\{c|Ac\leq h\}$, where  matrix $A\in \mathbb{R}^{2(n+1)\times (d+1)}$ and vector $h\in \mathbb{R}^{2(n+1)}$.  We now consider the following linear programming problem:
\begin{align}
    \text{LP: } \text{Max }&Z,\\
\text{s.t. }&Ac+eZ\leq h,\\
&Z\leq 0,
\end{align}
where $e\in \mathbb{R}^{2(n+1)}$ and $e^T=(1,1,\ldots,1)$. It is clear that $S\neq \emptyset$ if and only if the maximal value $Z^*=0$.
\end{proof}

According to Lemma \ref{Lemma-degree-II}, determining $\widetilde{\text{deg}}_{\varepsilon}(f)$ is equivalent to find out the minimal $d$ such that  $q(x)=\sum_{k=0}^d c_k V_k$ approximates $f$ with error $\varepsilon$.

\begin{theorem}
For any symmetrically  partial Boolean  function $f$ over $\{0,1\}^n$ with domain of definition $D$, and for the fixed $0\leq \varepsilon<1/2$, there exists an algorithm  to find out $\widetilde{\text{deg}}_{\varepsilon}(f)$ with time complexity $O(\log n)\cdot t(LP)$, where $t(LP)$ is the time complexity to use linear programming algorithm to  find the maximal value $Z^*$ in Lemma \ref{LP}.
\end{theorem}
\begin{proof}
%If $f$ is a constant function, then $\widetilde{\text{deg}}_{\varepsilon}(f)=0$.
 Let $\mathbf{b}=(b_0,b_1,\ldots,b_n)$ be the vector describing $f$.  Let subroutine $\text{LP}( n, \mathbf{b}, \varepsilon,d)=1 ~(0)$ if there does (not) exist polynomial $q$ with degree $d$ approximating $f$ with error $\varepsilon$. The subroutine $\text{LP}( n, \mathbf{b}, \varepsilon,d)$ can be done with a linear programming algorithm according to Lemma \ref{LP}.
We  give a binary search algorithm to find out $\widetilde{\text{deg}}_{\varepsilon}(f)$ as following:

\begin{algorithm}
\caption{Algorithm for finding out $\widetilde{\text{deg}}_{\varepsilon}(f)$}
\begin{algorithmic}[1]
\Procedure{Degree} {{\bf integer} $n$, {\bf array} $\mathbf{b}$,  {\bf real} $\varepsilon$}  \Comment{$\mathbf{b}\in\{0,1,*\}^{n+1}$}
\State {\bf integer} $l:=0, r:=n$;
\While{$l\leq r$}
\State $d=\lfloor(l+r)/2\rfloor$;
\If{$\text{LP}( n, \mathbf{b}, \varepsilon,d)$=0} {$l=d+1$};
\Else \ \  ${r=d-1}$;
\EndIf
\EndWhile\\
\Return $l$;
\EndProcedure
\end{algorithmic}
\end{algorithm}
In each iteration of the `while' loop,  it holds that $\text{LP}( n, \mathbf{b}, \varepsilon,r+1)=1$ and  $\text{LP}( n, \mathbf{b}, \varepsilon,l-1)=0$. We have $\widetilde{\text{deg}}_{\varepsilon}(f)\leq r+1$ and $\widetilde{\text{deg}}_{\varepsilon}(f)> l-1$. When the `while' loop is finished,  we have that $l=r+1$ and    $\widetilde{\text{deg}}_{\varepsilon}(f)\leq r+1=l$. Therefore,   $\widetilde{\text{deg}}_{\varepsilon}(f)= l$. The time complexity is $O(\log n)\cdot t(LP)$.
\end{proof}

\section{Generalized Deutsch-Jozsa problem}

In this section we  consider a generalized Deutsch-Jozsa problem $\text{ DJ}_n^k$ that was described by Eq. (\ref{DJF}), that is, the problem of distinguishing between the inputs of Hamming weight $n/2$ and Hamming weights in the set $\{0, 1,\ldots, k, n-k, n-k+1,\ldots,n\}$ for all even $n$ with $0\leq k<n/2$.

\subsection{Exact quantum algorithm }\label{Exact-Algorithm}

By combining the exact quantum query algorithms for the functions EXACT and THRESHOLD by Ambainis {\em et al} \cite{AIS13}, in this subsection we give an exact quantum query algorithm for computing $\text{ DJ}_n^k$.

\begin{theorem}\label{Th-upper-bound}
The exact quantum query complexity of $\text{\em DJ}_n^k$ satisfies:
\begin{equation}
    Q_E(\text{\em DJ}_n^k)\leq k+1.
\end{equation}
\end{theorem}
\begin{proof}
 We will give an exact quantum algorithm using $k+1$ queries for $\text{DJ}_n^k$.
One of the important subroutines that we will use in this paper is as following.
\begin{itemize}
  \item Input: $x=x_1,x_2,\ldots,x_m$.
  \item Output: If the output is $(0,0)$ then $|x|\neq m/2$. Otherwise, it will output $(i,j)$ such that $x_i\neq x_j$.
\end{itemize}
We call this subroutine $\mbox{Xquery}$. Let $x\in\{0,1\}^m$.  If $\mbox{Xquery}(m,x)=(0,0)$, then $|x|\neq m/2$. If  $\mbox{Xquery}(m,x)=(i,j)$, then $x_i\neq x_j$.

Indeed, according to \cite{AIS13} by Ambainis {\em et al}, the subroutine $\mbox{Xquery}$  can be implemented in one exact quantum query algorithm, and we put the details concerning the subroutine $\mbox{Xquery}$ in Appendix A.

Based on the subroutine $\mbox{Xquery}$, now we give an algorithm (Algorithm \ref{DJA}) for $\text{DJ}_n^k$. It is clear that Algorithm \ref{DJA} uses at most $k+1$ queries.

\begin{algorithm}
\caption{Algorithm for $\text{DJ}_n^k$}
\label{DJA}
\begin{algorithmic}[1]
 \Procedure{DJ} {{\bf integer} $n$, {\bf integer} $k$, {\bf array} $x$}\Comment{$x\in\{0,1\}^n$}
\State {\bf integer}  $l$:=1
\While{$l\leq k$}
\State Output $\gets$$\mbox{Xquery}(n,x)$
\If{Output=(0,0)} \Return 0
\EndIf
\If{Output=(i, j)}
\State $x\gets x\setminus\{x_i,x_j\}$
\State $l\gets l+1$
\State $n\gets n-2$
\EndIf
\EndWhile
\State Output $\gets$$\mbox{Xquery}(n,x)$
\If{Output=(0,0)} \Return 0
\EndIf
\If{Output=(i, j)}
\Return 1
\EndIf
\EndProcedure
\end{algorithmic}
\end{algorithm}

\end{proof}

\vskip 10mm

\subsection{Lower bound of exact quantum query complexity}\label{low-bound-MI}

 The purpose of this subsection is to prove that the exact quantum query complexity of $\text{DJ}_n^k$ is no less than $k+1$, i.e., $Q_E(\text{ DJ}_n^k)\geq k+1$.

\begin{theorem}\label{Th-lower-bound}
The exact quantum query complexity of $\text{\em DJ}_n^k$ satisfies:
\begin{equation}
    Q_E(\text{\em DJ}_n^k)\geq k+1.
\end{equation}
\end{theorem}

\begin{proof}
%It is clear that if $k=n/2-1$, then $\text{DJ}_n^k$ is a total function, which is  the function  $\text{EXACT}_n^{n/2}$. According to \cite{AISJ13}, $Q_E(\text{EXACT}_n^{n/2})=n/2=k+1$.
%According to Algorithm \ref{DJalgorithm}, we have the following fact:
%\begin{equation}
%    Q_E(\text{DJ}_n^k)\leq k+1.
%\end{equation}

We will prove that $\widetilde{\text{deg}}_{0}(\text{DJ}_n^k)\geq 2k+2$. Let us consider a simple case $k=1$ and $n\geq 6$ first. Suppose that $\widetilde{\text{deg}}_{0}(\text{DJ}_n^{1})\leq 3$, according to Lemma \ref{Lemma-degree}, there exists a multilinear polynomial $q(x)=\sum_{i=0}^{3}c_iV_i$ representing $\text{DJ}_n^{1}$. For $|x|=0$, we have $q(x)=c_0=f(x)=0$. For $|x|=1$, we have $q(x)=c_0+{1 \choose 1}c_1=f(x)=0$ and therefore $c_1=0$. For $|x|=n, n-1, n/2$, we have the following equations:
  \begin{equation}
\left\{\begin{array}{ll}
                    {n \choose 2}c_2+{n \choose 3}c_3=0,\\
                    {n-1\choose 2}c_2+{n-1\choose 3}c_3=0,\\
                    {n/2\choose 2}c_2+{n/2\choose 3}c_3=1.
                  \end{array}
 \right.
 \end{equation}
Let us consider the determinant
\begin{align}
   \left|
      \begin{array}{ccc}
        {n \choose 2} & {n \choose 3} \\
       {n-1\choose 2} & {n-1\choose 3} \\
      \end{array}
    \right|& =   \left|
      \begin{array}{ccc}
        {n \choose n-2} & {n \choose n-3} \\
       {n-1\choose n-3} & {n-1\choose n-4} \\
      \end{array}
    \right|\\
    &=\frac{1}{n}\left|
      \begin{array}{ccc}
        {n \choose n-2} & {n \choose n-3} \\
       n{n-1\choose n-3} & n{n-1\choose n-4} \\
      \end{array}
    \right|\\
    &=\frac{1}{n}\left|
      \begin{array}{ccc}
        {n \choose n-2} & {n \choose n-3} \\
       (n-2){n\choose n-2} & (n-3){n\choose n-3} \\
      \end{array}
    \right| \\
    &=\frac{1}{n}\left|
      \begin{array}{ccc}
        {n \choose n-2} & {n \choose n-3} \\
       {n\choose n-2} & 0 \\
      \end{array}
    \right|\neq 0.
\end{align}
Therefore, in order to satisfy the first two equations, we have $c_2=c_3=0$. The last equation will not hold, which means that such $q$ does not exist. Thus, $\widetilde{\text{deg}}_{0}(\text{DJ}_n^{1})\geq 4$. According the Lemma \ref{Lemma-lower-bound1}, we have $Q_E(\text{DJ}_n^{1})\geq \frac{1}{2}\widetilde{\text{deg}}_{0}(\text{DJ}_n^{1})\geq 2$. That is to say, the algorithm in Theorem \ref{Th-upper-bound} for  $\text{DJ}_n^{1}$ is optimal. The algorithm in \cite{MJM11} for $\text{DJ}_n^{1}$ is also optimal.

Now we consider for the general case.  Suppose that $\widetilde{\text{deg}}_{0}(\text{DJ}_n^{k})\leq 2k+1$, according to Lemma \ref{Lemma-degree}, there exists a multilinear polynomial $q(x)=\sum_{i=0}^{2k+1}c_iV_i$ representing $\text{DJ}_n^{k}$. For $0\leq |x|\leq k$, $f(x)=0$. Therefore, we have $c_0=c_1=\cdots=c_{k}=0$. For $|x|=n, n-1,\ldots, n-k$, we have the following equations:
  \begin{equation}
\left\{\begin{array}{ll}
                    {n \choose k+1}c_{k+1}+ {n \choose k+2}c_{k+2}+\cdots+{n \choose 2k+1}c_{2k+1}=0,\\
                    {n-1 \choose k+1}c_{k+1}+ {n-1 \choose k+2}c_{k+2}+\cdots+{n-1 \choose 2k+1}c_{2k+1}=0,\\
                    \cdots\\
                    {n-k \choose k+1}c_{k+1}+ {n-k \choose k+2}c_{k+2}+\cdots+{n-k \choose 2k+1}c_{2k+1}=0.\\
                  \end{array}
 \right.
 \end{equation}

Let us consider the determinant (see Appendix B for the detailed proof):
\begin{equation}\label{eq-det}
   \left|
      \begin{array}{cccc}
        {n \choose k+1} & {n \choose k+2}  & \cdots &  {n \choose 2k+1}\\
       {n-1 \choose k+1}& {n-1 \choose k+2}& \cdots & {n-1 \choose 2k+1}\\
       \vdots& \vdots& \ddots & \vdots\\
       {n-k \choose k+1} & {n-k \choose k+2} & \cdots &{n-k \choose 2k+1}\\
      \end{array}
    \right|=(-1)^{\frac{k(k+5)}{2}}\cdot\frac{\prod_{i=k+1}^{2k+1}{n \choose i}}{\prod_{i=1}^{k}{n \choose i}}\neq 0.
\end{equation}
Therefore, we have $c_{k+1}=\cdots=c_{2k+1}=0$.  Then for $|x|=n/2$, $f(x)=q(x)=0$, which is a contradiction. Therefore, $\widetilde{\text{deg}}_{0}(\text{DJ}_n^{k})\geq 2k+2$ and $Q_E(\text{DJ}_n^{k})\geq \frac{1}{2}\widetilde{\text{deg}}_{0}(\text{DJ}_n^{k})\geq k+1$.
\end{proof}

%\begin{equation}
%    a\overset{?}{=}b
%\end{equation}

%\begin{equation}
%    a\xlongequal{def}b
%\end{equation}

\subsection{Exact classical  query complexity}
\begin{theorem} \label{Th-classical-bound}
The classical deterministic  query complexity of $\text{\em DJ}_n^k$ satisfies:
\begin{equation}
    D(\text{\em DJ}_n^k)=n/2+k+1.
\end{equation}
\end{theorem}

\begin{proof}
 If the first $n/2$ queries return $x_i=1$ and the next $k$ queries return $x_i=0$, then we will need to make another query as well. Therefore,   $D(\text{DJ}_n^k)\geq n/2+k+1$.

Now suppose that we have made $n/2+k+1$ queries. If no more than $k$ queries return $x_i=0$,  then there are more than $n/2+1$ queries returning $x_i=1$ and $\text{DJ}_n^k(x)=0$.
 If no more than $k$ queries return $x_i=1$,  then there are more than $n/2+1$ queries returning $x_i=0$ and $\text{DJ}_n^k(x)=0$.  If there are  more than $k$ queries return $x_i=0$ and also more than $k$ queries return $x_i=1$, then it must be balanced and $\text{DJ}_n^k(x)=1$. Therefore,    $D(\text{DJ}_n^k)\leq n/2+k+1$ and the theorem has been proved.
\end{proof}

\begin{remark}

Again, we make some comparisons to the previous results. When $k=0$, this is the Deutsch-Jozsa problem; when $k=1$, this problem was considered by Montanaro {\em et al} \cite{MJM11} and an exact quantum 2-query algorithm was given to solve it, but the optimality was not verified. Also, the method in \cite{MJM11} is different (the unitary operator in their query algorithm was derived from distinguishing two orthogonal subsets of states).

\end{remark}

So far, according to  $\mathbf{Theorem ~\ref{Th-upper-bound} }$, $\mathbf{Theorem ~\ref{Th-lower-bound} }$, $\mathbf{Theorem ~\ref{Th-classical-bound} }$, our first main result, $\mathbf{Theorem ~\ref{QE(DJ)} }$  has been proved.

\section{symmetrically  partial functions with degree 1 or 2} \label{section-degree}

This section is to give all symmetrically  partial functions with degree 1 or 2 in the isomorphic sense.
  From now on, we just identity $\text{deg}(f)$ with $\widetilde{\text{deg}}_{0}(f)$  for any partial Boolean function $f$. %First, we consider the case of $n$ being odd. Indeed, the case of $n$ being even is similar.

\begin{lemma} \label{SPF121}
For $n>1$, then
\begin{equation}
\text{deg}(f^{(1)}_{n,n})=1,
\end{equation}
where $f^{(1)}_{n,n}$ is defined as Eq. (\ref{d1}) with $k=n$,
and the following symmetrically  partial  Boolean functions have degree  2:
\begin{align}
&f^{(1)}_{n,k}(x)=\left\{\begin{array}{ll}
                    0 &\text{\em if}\ |x|=0, \\
                    1 &\text{\em if}\ |x|= k,
                  \end{array}
 \right. \label{d1}   \\
&f^{(2)}_{n,k}(x)=\left\{\begin{array}{ll}
                    0 &\text{\em if}\ |x|=0, \\
                    1 &\text{\em if}\ |x|= k\ \text{\em or } |x|=k+1,
                  \end{array}
 \right.\\
&f^{(3)}_{n,l}(x)=\left\{\begin{array}{ll}
                    0 &\text{\em if}\ |x|=0 \ \text{\em or } |x|= n,\\
                    1 &\text{\em if}\ |x|= l,
                  \end{array}
 \right.\\
&f^{(4)}_{n}(x)=\left\{\begin{array}{ll}
                    0 &\text{\em if}\ |x|=0 \ \text{\em or } |x|= n,\\
                    1 &\text{\em if}\ |x|=\lfloor n/2\rfloor \ \text{\em or } |x|=\lceil n/2\rceil,
                  \end{array}
 \right.
 \end{align}
where $n-1\geq k\geq \lfloor n/2\rfloor$, and $\lceil n/2\rceil \geq l \geq \lfloor n/2\rfloor$. As usual, for odd $n$, $\lfloor n/2\rfloor=(n-1)/2$ and $\lceil n/2\rceil= (n+1)/2$.

\end{lemma}

\begin{proof}  Since the polynomial $q(x)=\frac{x_1+x_2+\ldots+x_n}{n}=\frac{1}{n}V_1$ can approximate $f^{(1)}_{n,n}$ with error $0$, $\text{deg}(f^{(1)}_{n,n})=1$ is verified. Next, we prove that the rest functions have degree 2 exactly.

Case 1. For $n-1\geq k\geq \lfloor n/2\rfloor$, it is easy to check that $\text{deg}(f^{(1)}_{n,k})> 1$, since the polynomial $q(x)=\frac{1}{k}V_1>1$ for $|x|\geq k$. As for the proof of $\text{deg}(f^{(1)}_{n,k})=2$, it follows from the following Case 2.

Case 2. For $n-1\geq k\geq \lfloor n/2\rfloor$, it is easy to verify that $\text{deg}(f^{(2)}_{n,k})>1$. Indeed, we can further verify that $\text{deg}(f^{(2)}_{n,k})=2$ in terms of the polynomial $q(x)=\frac{2}{k+1}V_1-\frac{2}{k(k+1)}V_2$.

 (1) First, we have $q(x)=f^{(2)}_{n,k}(x)$ for $|x|=0,k,k+1$.

 (2) Second, it follows $0\leq q(x)\leq 1$ from $q(x)=\frac{2}{k+1}$ for $|x|=1$ and $q(x)=\frac{2}{k+1}{i\choose 1}-\frac{2}{k(k+1)}{i \choose 2}=\frac{i(2k+1-i)}{k(k+1)}$ for $|x|=i\geq 2$. Indeed, since $k\geq \lfloor n/2\rfloor$, we have $q(x)\geq 0$; on the other hand, $i(2k+1-i)-k(k+1)=-(i-k-1/2)^2+1/4\leq -(1/2)^2+1/4\leq 0$, consequently we have $0\leq q(x)\leq 1$.

 So, this $q(x)$ can approximate $f^{(2)}_{n,k}$ with error $0$, and therefore $\text{deg}(f^{(2)}_{n,k})=2$.  Since $\text{deg}(f^{(1)}_{n,k})\leq \text{deg}(f^{(2)}_{n,k})$, we have also $\text{deg}(f^{(1)}_{n,k})=2$ for $n-1\geq k\geq \lfloor n/2\rfloor$.

Case 3. We now verify $\text{deg}(f^{(4)}_{n})=2$.
It is easy to verify that $\text{deg}(f^{(4)}_{n})>1$.
If $n$ is even, the function $f^{(4)}_{n}$ is
the well-known Deutsch-Jozsa problem with $Q_E(f^{(4)}_{n})=1$, and with Lemma \ref{Lemma-lower-bound1} we have $\text{deg}(f^{(4)}_{n})\leq 2$.   Therefore $\text{deg}(f^{(4)}_{n})=2$. If $n$ is odd with $n=2m+1$, then we consider the polynomial $p(x)=\frac{2}{m+1}V_1+\frac{2}{m(m+1)}V_2$. When $|x|=0,n,m,m+1$, we have $p(x)=f^{(4)}_{n}(x)$. If $|x|=1$, then $p(x)=\frac{2}{m+1}$ and $0\leq p(x)\leq 1$. For $2\leq |x|=i\leq n$, then  $p(x)=\frac{2}{m+1}{i\choose 1}-\frac{2}{m(m+1)}{i \choose 2}=\frac{i(2m+1-i)}{m(m+1)}$. According to Case 2 (2) above we also have $0\leq p(x)\leq 1$. Hence $p(x)$ can approximate $f^{(4)}_{n}(x)$ with error 0 and $\text{deg}(f^{(4)}_{n})= 2$.

Case 4.
It is easy to verify that $\text{deg}(f^{(3)}_{n,l})>1$. On the other hand, we always have  $\text{deg}(f^{(3)}_{n,l})\leq \text{deg}(f^{(4)}_{n})=2$. Therefore, we have $\text{deg}(f^{(3)}_{n,l})=2$.

First we note $\text{deg}(f^{(2)}_{n,l})\leq \text{deg}(f^{(3)}_{n,l})$ with $l=\lfloor n/2\rfloor$. Then we only need to give a polynomial $q(x)$ of degree 2 to approximate $f^{(3)}_{n,l}$ with error $0$. Consider this polynomial $q(x)=\frac{4}{n+1}V_1-\frac{8}{(n-1)(n+1)}V_2$. Omitting the details, we can check that $0\leq q(x)\leq 1$ and $q(x)=f^{(3)}_{n,l}(x)$ for $|x|=0,n,\lfloor n/2\rfloor,\lceil n/2\rceil$.

Summarily, the functions above have degree 2 except for the degree of $f^{(1)}_{n,n}$ being 1.

\end{proof}

Indeed, the following lemma shows that those functions in Lemma \ref{SPF121} contain all symmetrically  partial  functions with degree 1 or 2. First, we consider the case of $n$ being odd. Indeed, the case of $n$ being even is similar.

\begin{lemma}\label{SPF122}
Let $n>1$ be odd, and let $f:\{0,1\}^n \to\{0,1\}$  be an $n$-bit symmetrically  partial  Boolean function. Then:

  (1) if $\text{\em deg}(f)=1$, then  $f$ is isomorphic to the function $f^{(1)}_{n,n}$;

  (2) if $\text{\em deg}(f)=2$, then  $f$ is isomorphic to one of the functions
\begin{align}
&f^{(1)}_{n,k}(x)=\left\{\begin{array}{ll}
                    0 &\text{\em if}\ |x|=0, \\
                    1 &\text{\em if}\ |x|= k,
                  \end{array}
 \right.\\
&f^{(2)}_{n,k}(x)=\left\{\begin{array}{ll}
                    0 &\text{\em if}\ |x|=0, \\
                    1 &\text{\em if}\ |x|= k\ \text{\em or } |x|=k+1,
                  \end{array}
 \right.\\
 &f^{(3)}_{n,l}(x)=\left\{\begin{array}{ll}
                    0 &\text{\em if}\ |x|=0 \ \text{\em or } |x|= n,\\
                    1 &\text{\em if}\ |x|= l,
                  \end{array}
 \right.\\
&f^{(4)}_{n,1}(x)=\left\{\begin{array}{ll}
                    0 &\text{\em if}\ |x|=0 \ \text{\em or } |x|= n,\\
                    1 &\text{\em if}\ |x|=\lfloor n/2\rfloor \ \text{\em or } |x|=\lceil n/2\rceil,
                  \end{array}
 \right.
 \end{align}
where $n-1\geq k\geq \lfloor n/2\rfloor$, and $\lceil n/2\rceil \geq l\geq \lfloor n/2\rfloor$.

\end{lemma}

\begin{proof}

The lemma can be easily verified  for $n\leq 3$, so we now prove the case of $n>3$. Since the degree of function $f$ to be considered is not $0$, $f$ is not a constant function. Let $(b_0,b_1,\ldots,b_n)\in\{0,1,*\}^{n+1}$ be the vector describing $f$. Then there exist $0\leq i<j\leq n$ such that $b_i=0$ and $b_j=1$,  otherwise, we can consider its isomorphic function $(\bar{b}_0,\bar{b}_1,\ldots,\bar{b}_n)$, instead of $f$. Also, we note that if a polynomial $q$ approximates  function $(b_0,b_1,\ldots,b_n)$ with error 0, then the polynomial $1-q$ can approximate  function $(\bar{b}_0,\bar{b}_1,\ldots,\bar{b}_n)$ with error $0$, and the polynomial $q(\bar{x})$  can approximate  function $(b_n, b_{n-1}, \ldots, b_0)$  with error $0$, as well as the polynomial $1-q(\bar{x})$  can approximate  function $(\bar{b}_n, \bar{b}_{n-1}, \ldots, \bar{b}_0)$  with error $0$.

(1)  If $ \text{ deg}(f) =1$, according to Lemma \ref{Lemma-degree-II}, there is a polynomial $q(x)=c_0+c_1V_1$ that approximates $f$ with error 0. Therefore, we have:

 (a) $q(x)=c_0+c_1\cdot i=0$ for $|x|=i$, and $q(x)=c_0+c_1\cdot j=1$ for $|x|=j$;

  (b)  $0\leq q(x)=c_0+c_1\cdot m\leq 1$ for $0\leq |x|=m\leq n$.

  From (a), we have $c_0=-i/(j-i)$ and $c_1=1/(j-i)$.  However,  (b) implies $0\leq c_0\leq 1$ by taking  $|x|=0$. Therefore $i=0$ and $c_0=0$. In addition, when $|x|=n$, it follows from (b) that $c_1\leq 1/n$. Therefore $j=n$ and $c_1= 1/n$. Moreover, $b_k$ is undefined (i.e., $*$) for $k\not\in \{0,n\}$, otherwise, $q(x)=k/n\neq b_k$.

  Therefore, the symmetrically  partial  Boolean function $f$ must isomorphic to $(0,*,\ldots,*,1)$, i.e. the function $f^{(1)}_{n,n}$.

(2) If $ \text{deg}(f) =2$, with Lemma \ref{Lemma-degree-II}, there is a polynomial $q(x)=c_0+c_1V_1+c_2V_2$ approximating $f$ with error $0$. Suppose that $b_i=0$ and $b_j=1$ for some $0<i<j<n$.
Then we have:

(a) $q(x)=c_0+c_1{i \choose 1}+c_2{i \choose 2}=0$ for $|x|=i$, and $q(x)=c_0+c_1{j \choose 1}+c_2{j \choose 2}=1$ for $|x|=j$;

  (b) $0\leq q(x)=c_0+c_1{m \choose 1}+c_2{m\choose 2}\leq 1$ for $0\leq |x|=m\leq n$.

 By virtue of (a), we have:

  $$c_1=-\frac{i+j-1}{ij}c_0-\frac{i-1}{j(j-i)},$$

  $$c_2=\frac{2}{ij}c_0+\frac{2}{j(j-i)}.$$

  On the other hand, by taking $|x|=0$, (b) implies $0\leq c_0\leq 1$. Now for $|x|=n$ we have
  \begin{align}
q(x)&=c_0+c_1{n \choose 1}+c_2{n \choose 2}\nonumber\\
&=c_0+n(-\frac{i+j-1}{ij}c_0-\frac{i-1}{j(j-i)})+\frac{1}{2}n(n-1)(\frac{2}{ij}c_0+\frac{2}{j(j-i)})\\
&=\frac{(n-i)(n-j)}{ij}c_0+\frac{n(n-i)}{j(j-i)}\\
&\geq \frac{n(n-i)}{j(j-i)}\\
&>1,
 \end{align}
but this contradicts to  $q(x)\leq 1$. This contradiction is derived from the assumption of $b_i=0$ and $b_j=1$ for some $0<i<j<n$. So, by combining with the isomorphic property, we have obtained the following result:

{\it Result 1: If $b_i\in\{0,1\}$ for $0<i<n$, then $b_j\neq \bar{b}_i$ for  $0<j<n$.   }

So, furthermore it suffices to consider $b_0=0$ or $b_n=1$.
We now consider the case $b_0=0$.

Note that $b_0=0$ implies  $c_0=0$.
Suppose that  $b_1=1$. Then we have (1) $q(x)=c_1{1 \choose 1}+c_2{1 \choose 2}=c_1=1$ for $|x|=1$ and, (2) $0\leq q(x)=c_1{m \choose 1}+c_2{m\choose 2}\leq 1$ for $0\leq |x|=m\leq n$.

With (2), by taking $|x|=n$, we have $0\leq c_1{n \choose 1}+c_2{n \choose 2}\leq 1$ and therefore $\frac{-2}{n-1}\leq c_2\leq \frac{-2}{n}$.
 When $|x|=n-2$, with (2) we have
 \begin{align}
q(x)=&c_1{n-2 \choose 1}+c_2{n-2 \choose 2}\\
=& n-2+\frac{1}{2}(n-2)(n-3)c_2\\
\geq & n-2+\frac{1}{2}(n-2)(n-3)\frac{-2}{n-1}\\
>&1,
\end{align}
which is a contradiction to $q(x)\leq 1$. Therefore $b_1\neq 1$.

Indeed, for $1<j< \lfloor n/2\rfloor  $ we also have $b_j\neq 1$:  Suppose that $b_j=1$ for some $1<j<\lfloor n/2\rfloor$. Then we  have

(a) $q(x)=c_1{j \choose 1}+c_2{j\choose 2}=1$ for $|x|=j$;

 (b) $0\leq q(x)=c_1{m \choose 1}+c_2{m\choose 2}\leq 1$ for $0\leq |x|=m\leq n$.

 With (a) we have
 \begin{equation}\label{c2}
 c_2=\frac{2}{j(j-1)}-\frac{2}{j-1}c_1.
  \end{equation}
  With (b), by taking $|x|=n$ and combining with Eq. (\ref{c2}) we have
  \begin{align}
0&\leq nc_1+\frac{1}{2}n(n-1)c_2\\
&= nc_1+\frac{1}{2}n(n-1)(\frac{2}{j(j-1)}-\frac{2}{j-1}c_1)\\
&\leq 1,
\end{align}
which follows
  \begin{equation}
 \frac{n+j-1}{nj}\leq c_1\leq \frac{n-1}{j(n-j)}.
  \end{equation}

  Similarly, with (b), by taking $|x|=\lceil n/2\rceil$ and combining with Eq. (\ref{c2}) we have
  \begin{align}
q(x)&=\frac{n+1}{2} c_1+\frac{1}{2} \frac{n+1}{2} ( \frac{n+1}{2}-1)c_2\\
&=\frac{n+1}{2} c_1+\frac{1}{2}\frac{n+1}{2} (\frac{n+1}{2}-1)(\frac{2}{j(j-1)}-\frac{2}{j-1}c_1)\\
&=\frac{\frac{n+1}{2} (\frac{n+1}{2}-1)}{j(j-1)}+(\frac{n+1}{2}-\frac{\frac{n+1}{2} (\frac{n+1}{2}-1)}{j-1})c_1.
\end{align}
 Therefore, with  $j<\lceil n/2\rceil$, we have
   \begin{align}
q(x)&\geq \frac{\frac{n+1}{2} (\frac{n+1}{2}-1)}{j(j-1)}+(\frac{n+1}{2}-\frac{\frac{n+1}{2} (\frac{n+1}{2}-1)}{j-1})\frac{n-1}{j(n-j)}\\
&=\frac{\frac{n+1}{2}(n-\frac{n+1}{2})}{j(n-j)}.
\end{align}
 Furthermore, with  $j<\lfloor n/2\rfloor$, we have
 \begin{align}
&\frac{n-1}{2}(n-\frac{n-1}{2})-j(n-j)\\
&=(\frac{n-1}{2}-j)(n-\frac{n-1}{2}-j)>0.
\end{align}
 Together with the above equations we obtain $q(x)\geq 1$ for $|x|=\lceil n/2\rceil$, a contradiction. So, we conclude the following result.

  {\it Result 2: If $b_j=1$, then $j\geq \lfloor n/2\rfloor$.}

Let $l$ be the smallest integer satisfying $b_l=1$ (of course, $l\geq \lfloor n/2\rfloor$). Next we will prove that $b_k\neq 1$ for any $k>l+1$.

Assume that $b_k=1$. Then we have $c_0=0$, $c_0+c_1{l \choose 1}+c_2{l \choose 2}=1$ and $c_0+c_1{k \choose 1}+c_2{k \choose 2}=1$. Therefore, $c_0=0$, $c_1=\frac{l+k-1}{lk}$ and $c_2=\frac{-2}{lk}$. For $|x|=l+1<k$, we have
 \begin{align}
q(x)&=c_0+c_1{l+1 \choose 1}+c_2{l+1 \choose 2}\\
&=(l+1)\frac{l+k-1}{lk}+\frac{1}{2}l(l+1)\cdot \frac{-2}{lk}\\
&=\frac{(l+1)(k-1)}{lk}\\
&>1,
\end{align}
 which is a contradiction. Therefore, we have proved this result:

 {\it Result 3: If $l$ is the smallest integer satisfying $b_l=1$, then $b_k\neq 1$ for any $k>l+1$.}

Now let $l$ be the smallest integer satisfying $b_l=1$. Next we complete the proof by considering $l$ with five cases.

(I) $l=n$, i.e., $b_i\neq 1$ for $1\leq i\leq n-1$. Since the isomorphic function $(\bar{b}_n,\bar{b}_{n-1},\ldots,\bar{b}_0)$ also has degree 2 and $\bar{b}_n=0, \bar{b}_0=1$, according to the above Result 3, we have $\bar{b}_i\neq 1$ for $2\leq i\leq n-1$, that is $b_i\neq 0$ for $2\leq i\leq n-1$. Therefore $f$ can be described by $\mathbf{b}=(0,b_1,*,\ldots,*,1)$. If $b_1=*$, then $\text{deg}(f)=1$. So, we conclude $f$ must be the formulation $\mathbf{b}=(0,0,*,\ldots,*,1)$, and it is isomorphic to $f^{(2)}_{n,n-1}$.

(II) $l=n-1$, i.e., $b_i\neq 1$ for $1\leq i<n-1$. From Result 1 we know that $b_j\neq \bar{b}_l$ for $0<j<n$, that is $b_j\neq 0$ for $0<j<n$. As a result, $f$ has this formulation $\mathbf{b}=(0,*,\ldots,*,1,b_n)$. However, if $b_n=0$, then it is isomorphic to $(0,1,*,\ldots,0)$, which results in $\text{deg}(f)>2$. Therefore, $b_n=1$ and $f$ is described by $\mathbf{b}=(0,*,\ldots,*,1,1)$, exactly the function $f^{(2)}_{n,n-1}$ as well.

 (III) $\lceil n/2\rceil<l<n-1$, that is, $b_i\neq 1$ for $1\leq i<l$ and $l+1<i\leq n$. From Result 1 it follows that $b_j\neq 0$ for $0<j<n$. As a consequence, $f$ has the formulation $\mathbf{b}=(0,*,\ldots,*,1,b_{l+1},*,\ldots,*,b_n)$ where $b_n\neq 1$. If $b_n=0$, then $f$ is isomorphic to $(0,*,\ldots,*,b_{l+1},1,*,\ldots,*,0)$, which from Result 2 follows $\text{deg}(f)>2$, a contradiction. So, it holds that $b_n=*$, and $f$ thus has the form $\mathbf{b}=(0,*,\ldots,*,1,b_{l+1},*,\ldots,*)$. In this representation, $b_{l+1}=*$ implies the function $f=f^{(1)}_{n,l}$; and $b_{l+1}=1$ results in $f=f^{(2)}_{n,l}$.

 (IV) $l=\lceil n/2\rceil$. Then $b_i\neq 1$ for $1\leq i<(n+1)/2$ and $(n+1)/2+1<i\leq n$.  With Result 1 we have also $b_j\neq 0$ for $0<j<n$. Now it concludes that $f$ has this representation $\mathbf{b}=(0,*,\ldots,*,1,b_{l+1},*,\ldots,b_n)$. Furthermore, if $b_{l+1}=*$ and $b_n=*$, then $f=f^{(1)}_{n,l}$; if $b_{l+1}=*$ and $b_n=0$, then  $f=f^{(3)}_{n,l}$;
if $b_{l+1}=1$ and $b_n=*$, then $f=f^{(2)}_{n,l}$.  Finally, both $b_{l+1}=1$ and $b_n=0$ result in $\text{deg}(f)>2$, which is an impossible case.

(V) $l=\lfloor n/2\rfloor$. Then  $b_i\neq 1$ for $1\leq i<(n-1)/2$ and $(n-1)/2+1<i\leq n$.  Result 1 also implies $b_j\neq 0$ for $0<j<n$.  Therefore $f$ has this representation $\mathbf{b}=(0,*,\ldots,*,1,b_{l+1},*,\ldots,b_n)$. In addition, if $b_{l+1}=*$ and $b_n=*$, then $f=f^{(1)}_{n,l}$; if $b_{l+1}=*$ and $b_n=0$, then  $f=f^{(3)}_{n,l}$;
if $b_{l+1}=1$ and $b_n=*$, then  $f=f^{(2)}_{n,l}$;  if $b_{l+1}=1$ and $b_n=0$, then  $f=f^{(4)}_{n,1}$.

So, the case of $b_0=0$ has been proved. Finally, we consider the case $b_n=1$ with an isomorphic method. Because $\mathbf{b}=(b_0,b_1,\ldots,b_n)$ is isomorphic to $\mathbf{b}'=(\bar{b}_n,\bar{b}_{n-1},\ldots,\bar{b}_0)$, we have $\text{deg}(\mathbf{b}')=\text{deg}(\mathbf{b})=2$. Since $\bar{b}_n=0$, as we have proved above,  the function $\mathbf{b}'$ must be isomorphic to the one of the following functions: $f^{(1)}_{n,k}$, $f^{(2)}_{n,k}$, $f^{(3)}_{n,l}$ and $f^{(4)}_{n,1}(x)$, where $n\geq k\geq \lfloor n/2\rfloor $ and $\lfloor n/2\rfloor \leq l\leq \lceil n/2\rceil $. Therefore, $f$  must be isomorphic to the one of the following functions: $f^{(1)}_{n,k}$, $f^{(2)}_{n,k}$, $f^{(3)}_{n,l}$ and $f^{(4)}_{n,1}(x)$, where $n\geq k\geq \lfloor n/2\rfloor $ and $\lfloor n/2\rfloor \leq l\leq \lceil n/2\rceil $.

\end{proof}

If $n$ is an even, then with a similar process of proof to the case of $n$ being odd we have the following result.

\begin{lemma}\label{SPF123}
Let $n>1$ be even, and let $f:\{0,1\}^n \to\{0,1\}$  be an $n$-bit symmetrically  partial Boolean function. Then:

  (1) if $\text{\em deg}(f)=1$,  then $f$ is isomorphic to the function $f^{(1)}_{n,n}$;

  (2) if $\text{\em deg}(f)=2$,  then is isomorphic to one of the functions
\begin{align}
&f^{(1)}_{n,k}(x)=\left\{\begin{array}{ll}
                    0 &\text{\em if}\ |x|=0, \\
                    1 &\text{\em if}\ |x|= k,
                  \end{array}
 \right.\\
&f^{(2)}_{n,k}(x)=\left\{\begin{array}{ll}
                    0 &\text{\em if}\ |x|=0, \\
                    1 &\text{\em if}\ |x|= k\ \text{\em or } |x|=k+1,
                  \end{array}
 \right.\\
 &f^{(3)}_{n,n/2}(x)=\left\{\begin{array}{ll}
                    0 &\text{\em if}\ |x|=0 \ \text{\em or } |x|= n,\\
                    1 &\text{\em if}\ |x|= n/2,
                  \end{array}
 \right.
 \end{align}
where $n-1\geq k\geq n/2$.

\end{lemma}

Therefore, combining Lemmas \ref{SPF121} with \ref{SPF122} and \ref{SPF123}, we have the following two lemmas.

\begin{lemma} \label{odd}
Let $n>1$ be odd, and let $f:\{0,1\}^n \to\{0,1\}$  be an $n$-bit symmetrically  partial  Boolean function. Then:

  (1) $\text{\em deg}(f)=1$ if and only if $f$ is isomorphic to the function $f^{(1)}_{n,n}$;

  (2) $\text{\em deg}(f)=2$ if and only if $f$ is isomorphic to one of the functions
\begin{align}
&f^{(1)}_{n,k}(x)=\left\{\begin{array}{ll}
                    0 &\text{\em if}\ |x|=0, \\
                    1 &\text{\em if}\ |x|= k,
                  \end{array}
 \right.\\
&f^{(2)}_{n,k}(x)=\left\{\begin{array}{ll}
                    0 &\text{\em if}\ |x|=0, \\
                    1 &\text{\em if}\ |x|= k\ \text{\em or } |x|=k+1,
                  \end{array}
 \right.\\
 &f^{(3)}_{n,l}(x)=\left\{\begin{array}{ll}
                    0 &\text{\em if}\ |x|=0 \ \text{\em or } |x|= n,\\
                    1 &\text{\em if}\ |x|= l,
                  \end{array}
 \right.\\
&f^{(4)}_{n,1}(x)=\left\{\begin{array}{ll}
                    0 &\text{\em if}\ |x|=0 \ \text{\em or } |x|= n,\\
                    1 &\text{\em if}\ |x|=\lfloor n/2\rfloor \ \text{\em or } |x|=\lceil n/2\rceil,
                  \end{array}
 \right.
 \end{align}
where $n-1\geq k\geq \lfloor n/2\rfloor$, and $\lceil n/2\rceil \geq l\geq \lfloor n/2\rfloor$.

\end{lemma}

\begin{lemma}\label{even}
Let $n>1$ be even, and let $f:\{0,1\}^n \to\{0,1\}$  be an $n$-bit symmetrically  partial Boolean function. Then:

  (1) $\text{\em deg}(f)=1$ if and only if $f$ is isomorphic to the function $f^{(1)}_{n,n}$;

  (2) $\text{\em deg}(f)=2$ if and only if $f$ is isomorphic to one of the functions
\begin{align}
&f^{(1)}_{n,k}(x)=\left\{\begin{array}{ll}
                    0 &\text{\em if}\ |x|=0, \\
                    1 &\text{\em if}\ |x|= k,
                  \end{array}
 \right.\\
&f^{(2)}_{n,k}(x)=\left\{\begin{array}{ll}
                    0 &\text{\em if}\ |x|=0, \\
                    1 &\text{\em if}\ |x|= k\ \text{\em or } |x|=k+1,
                  \end{array}
 \right.\\
 &f^{(3)}_{n,n/2}(x)=\left\{\begin{array}{ll}
                    0 &\text{\em if}\ |x|=0 \ \text{\em or } |x|= n,\\
                    1 &\text{\em if}\ |x|= n/2,
                  \end{array}
 \right.
 \end{align}
where $n-1\geq k\geq n/2$.

\end{lemma}

Combining Lemmas \ref{odd} and \ref{even} we obtain the following result concerning the characterizations of all symmetrically  partial Boolean functions with degree 1 or 2.

\begin{theorem}\label{T-deg-less2}

Let $n>1$  and let $f:\{0,1\}^n \to\{0,1\}$  be an $n$-bit symmetrically  partial  Boolean function. Then:

  (1) $\text{\em deg}(f)=1$ if and only if $f$ is isomorphic to the function $f^{(1)}_{n,n}$;

  (2) $\text{\em deg}(f)=2$ if and only if $f$ is isomorphic to one of the functions
\begin{align}
&f^{(1)}_{n,k}(x)=\left\{\begin{array}{ll}
                    0 &\text{\em if}\ |x|=0, \\
                    1 &\text{\em if}\ |x|= k,
                  \end{array}
 \right.\\
&f^{(2)}_{n,k}(x)=\left\{\begin{array}{ll}
                    0 &\text{\em if}\ |x|=0, \\
                    1 &\text{\em if}\ |x|= k\ \text{\em or } |x|=k+1,
                  \end{array}
 \right.\\
 &f^{(3)}_{n,l}(x)=\left\{\begin{array}{ll}
                    0 &\text{\em if}\ |x|=0 \ \text{\em or } |x|= n,\\
                    1 &\text{\em if}\ |x|= l,
                  \end{array}
 \right.\\
&f^{(4)}_{n}(x)=\left\{\begin{array}{ll}
                    0 &\text{\em if}\ |x|=0 \ \text{\em or } |x|= n,\\
                    1 &\text{\em if}\ |x|=\lfloor n/2\rfloor \ \text{\em or } |x|=\lceil n/2\rceil,
                  \end{array}
 \right.
 \end{align}
where $n-1\geq k\geq \lfloor n/2\rfloor$, and $\lceil n/2\rceil\geq l\geq \lfloor n/2\rfloor$.

\end{theorem}

\begin{remark}
According to Lemma \ref{Lemma-lower-bound1}, symmetric Boolean functions that can be computed by exact quantum 1-query  algorithm must have degree not more than 2. Therefore, Theorem \ref{T-deg-less2} describes all possible symmetric Boolean functions that can be computed by exact quantum 1-query algorithms.
\end{remark}

\section{symmetrically partial functions with exact quantum 1-query complexity}

In this section, we try to find out all symmetrically  partial  functions that can be computed with exact quantum 1-query algorithms. More precisely, we will obtain that any partial symmetric  function has exact quantum 1-query complexity if and only if it can be computed by Deutsch-Jozsa algorithm.

First, we have the following proposition that was proved in \cite{GQZ15}.
\begin{proposition}\cite{GQZ15} \label{GQZ15}
Let $n>1$. Then for any $n\geq k\geq \lceil n/2\rceil$,  $Q_E(f^{(1)}_{n,k})=1$. %When $n$ is even, $f3_n^l$ and $f4_n$ are the same function, which is the well known Deutsch-Jozsa promise problem and $Q_E(f4_n)=1$.
 %When $n=2$, the function $f2_3^1$ is the $\text{OR}_2$ function and $Q_E(f2_2^1)=2$. When $n=3$, the function $f4_3$ is the  $\text{NAE}_3$ function in \cite{MJM15},  which is equivalent to $(x_1\oplus x_2) \vee (x_1 \oplus x_3)$ and $Q_E(f4_3)=2$.
\end{proposition}

The exact quantum query complexity of $f^{(2)}_{n,k}$ is beyond 1, and this is the following result.

\begin{theorem}\label{T-k2}
Let $n>1$. Then for any $0<k<n$, $Q_E(f^{(2)}_{n,k})\geq 2$.

\end{theorem}

\begin{proof} The proof is divided into two cases in terms of $0<k<n-1$ and $k=n-1$.

{\bf Case 1: $0<k<n-1$.}
Assume that there is an exact quantum 1-query algorithm with $U_0$, $O_x$ and $U_1$ being the sequence of unitary operators for $f^{(2)}_{n,k}$, and $|\psi_s\rangle$ being its starting state.
Let $U_0|\psi_s\rangle=\sum_{i=0,j=1}^{n,m}\alpha_{ij}|i\rangle|j\rangle$.
When $|x|=0$, we have
\begin{equation}
|\psi_0\rangle= O_xU_0|\psi_s\rangle=\sum_{i=0,j=1}^{n,m}\alpha_{ij}|i\rangle|j\rangle. \nonumber
\end{equation}
Denote $\beta_i=\sum_{j=1}^m |\alpha_{ij}|^2$.

 When $|x|=k$, let $x_1=\cdots=x_{k}=1$ and $x_{k+1}=\cdots=x_n=0$,  and then we have
\begin{align}
 |\psi_k\rangle &= O_xU_0|\psi_s\rangle \nonumber\\
 &=\sum_{i=1,j=1}^{k,m}-\alpha_{ij}|i\rangle|j\rangle+\sum_{j=1}^{m}\alpha_{0j}|0\rangle|j\rangle+\sum_{i=k+1,j=1}^{n,m}\alpha_{ij}|i\rangle|j\rangle. \nonumber \end{align}
 Since the algorithm is exact, the quantum state $U_1|\psi_0\rangle$ must be orthogonal to the quantum state $U_1|\psi_k\rangle$. Therefore, we have
\begin{align}
0=&(U_1|\psi_0\rangle)^{\dag}U_1|\psi_k\rangle = \langle \psi_0|\psi_k\rangle\\
=&\left(\sum_{i=0,j=1}^{n,m}\overline{\alpha}_{ij}\langle i|\langle j|\right) \times\\ &\left(\sum_{j=1}^{m}\alpha_{0j}|0\rangle|j\rangle+\sum_{i=1,j=1}^{k,m}-\alpha_{ij}|i\rangle|j\rangle+\sum_{i=k+1,j=1}^{n,m}\alpha_{ij}|i\rangle|j\rangle\right)\\
=&\sum_{j=1}^{m}|\alpha_{0j}|^2+\sum_{i=1,j=1}^{k,m}-|\alpha_{ij}|^2+\sum_{i=k+1,j=1}^{n,m}|\alpha_{ij}|^2\\
=&\beta_0-\sum_{i=1}^k \beta_i+\sum_{i=k+1}^n \beta_i. \label{k0}
\end{align}

When $|x|=k+1$, let $x_1=\cdots=x_{k+1}=1$ and $x_{k+2}=\cdots=x_n=0$, and then we have
\begin{align}
|\psi_{(k+1)^0}\rangle &=O_xU_0|\psi_s\rangle  \nonumber \\
&=
\sum_{i=1,j=1}^{k+1,m}-\alpha_{ij}|i\rangle|j\rangle+\sum_{j=1}^{m}\alpha_{0j}|0\rangle|j\rangle+
\sum_{i=k+2,j=1}^{n,m}\alpha_{ij}|i\rangle|j\rangle. \nonumber
\end{align}
We also have
\begin{align}
0=&(U_1|\psi_0\rangle)^{\dag}U_1|\psi_{(k+1)^0}\rangle=\langle \psi_0|\psi_{(k+1)^0}\rangle\\
=&\left(\sum_{i=0,j=1}^{n,m}\overline{\alpha}_{ij}\langle i|\langle j|\right)\times \\ &\left(\sum_{j=1}^{m}\alpha_{0j}|0\rangle|j\rangle+\sum_{i=1,j=1}^{k+1,m}-\alpha_{ij}|i\rangle|j\rangle+\sum_{i=k+2,j=1}^{n,m}\alpha_{ij}|i\rangle|j\rangle\right)\\
=&\beta_0 -\sum_{i=1}^{k+1}\beta_i+ \sum_{i=k+2}^{n} \beta_i \\
=&-2\beta_{k+1},
\end{align}
where the last equality is according to $\beta_0-\sum_{i=1}^k \beta_i+\sum_{i=k+1}^n \beta_i=0$ from Eq. (\ref{k0}).
So, we have $\beta_{k+1}=0$.

Let $x_1=\cdots= x_k=1$, $x_l=1$ for an $l>k+1$ and let the others be 0. With such an input $x$ then we can similarly  obtain $\beta_{l}=0$.

As a result, we have obtained that $\beta_{i}=0$ for $k+1\leq i\leq n$, and therefore
\begin{equation}
\beta_0-\sum_{i=1}^{k}\beta_i=0. \label{beta0}
\end{equation}

Let $x_l=0$ for an $l\leq k$, let $x_j=1$ for $j\leq k+2$ with $j\neq l$, and the others $x_i=0$ for $n\geq i\geq k+3$. For such  an input $x$, denote $|\psi_{(k+1)^l}\rangle=O_xU_0|\psi_s\rangle$. Then similarly we have
 \begin{align}
0=&(U_1|\psi_0\rangle)^{\dag}U_1|\psi_{(k+1)^l}\rangle=\langle \psi_0|\psi_{(k+1)^l}\rangle\\
=&\left(\sum_{i=0,j=1}^{n,m}\overline{\alpha}_{ij}\langle i|\langle j|\right) \times\\ &\left(\sum_{j=1}^{m}\alpha_{0j}|0\rangle|j\rangle+\sum_{i=1,j=1}^{l-1,m}-\alpha_{ij}|i\rangle|j\rangle + \sum_{j=1}^{m}\alpha_{lj}|l\rangle|j\rangle +\sum_{i=l+1,j=1}^{k+2,m}-\alpha_{ij}|i\rangle|j\rangle\right)\\
=&\beta_0-\sum_{i=1}^{l-1}\beta_i+\beta_l-\sum_{i=l+1}^{k+2}\beta_i\\
=&2\beta_l,
%=&\left(\sum_{j=1}^{m}|\alpha_{0j}|^2+\sum_{i=1,j=1}^{k,m}-|\alpha_{ij}|^2\right)+2\sum_{j=1}^{m}|\alpha_{lj}|^2=2\sum_{j=1}^{m}|\alpha_{lj}|^2.
\end{align}
where the last equality follows from  Eq. (\ref{beta0}).

Therefore, now we  have $\beta_l=0$  for $1\leq l\leq n$. From  Eq. (\ref{beta0}) it follows that
 $\beta_0=0$. So far, we have concluded that $\beta_l=0$ for $0\leq l\leq n$, which result in
 $\alpha_{ij}=0$ for $0\leq i\leq n$ and $1\leq j \leq m$.

Consequently,  $U_0|\psi_s\rangle=\sum_{i=0,j=1}^{n,m}\alpha_{ij}|i\rangle|j\rangle=\mathbf{0}$ and $|\psi_s\rangle=\mathbf{0}$,  a contradiction.

{\bf Case 2: $k=n-1$.}
By using the above method as Case 1, it is easy to verify that 1-query is not enough.
Therefore, we have $Q_E(f^{(2)}_{n,k})\geq 2$ for $0<k<n$.
\end{proof}

Combining Lemma \ref{even} and Theorem \ref{T-k2} as well as Proposition \ref{GQZ15}, we have the following result.
\begin{theorem}\label{even-1-query}
Let $n$ be even  and let $f:\{0,1\}^n \to\{0,1\}$  be an $n$-bit symmetrically  partial function. Then $Q_E(f)=1$ if and only if $f$ is isomorphic to one of these functions: $f^{(1)}_{n,k}$ and $f^{(3)}_{n,n/2}$,
 where $ k\geq n/2$.
\end{theorem}

\begin{proof}
Suppose that $Q_E(f)=1$.  Then according to Lemma \ref{Lemma-lower-bound1}, $\text{deg}(f)\leq 2Q_E(f)=2$. By virtue of Lemma \ref{even}, $f$ is isomorphic to one of these functions:
$f^{(1)}_{n,n}, f^{(1)}_{n,k}, f^{(2)}_{n,k}, f^{(3)}_{n,n/2}$ for $n-1\geq k\geq n/2$. Furthermore, Proposition \ref{GQZ15} shows that $Q_E(f^{(1)}_{n,k})=1$ for $n\geq k\geq n/2$; and $Q_E(f^{(3)}_{n,n/2})=1$  is derived from Deutsch-Jozsa algorithm; Theorem \ref{T-k2} gives $Q_E(f^{(2)}_{n,k})\geq 2$ for any $0<k<n$.
Consequently,  $Q_E(f)=1$ implies that  $f$ is isomorphic to one of these functions: $f^{(1)}_{n,k}$ and $f^{(3)}_{n,n/2}$,
 where $ k\geq n/2$.

 On the other hand, since $Q_E(f^{(1)}_{n,k})=1$ for $n\geq k\geq n/2$, and $Q_E(f^{(3)}_{n,n/2})=1$,
 if  $f$ is isomorphic to one of these functions: $f^{(1)}_{n,k}$ and $f^{(3)}_{n,n/2}$,
 where $ k\geq n/2$,  by {\bf Fact \ref{fact1}}, $Q_E(f)=1$ follows.

\end{proof}

To consider the case of $n$ being odd, we need the following result.
\begin{theorem}\label{T-floor-n/2}
 For any integer $h>0$, $Q_E(f^{(1)}_{2h+1,h})\geq 2$.
\end{theorem}
\begin{proof}

The method of proof is similar to that of Theorem \ref{T-k2}.
Let $n=2h+1$.
Assume that there is an exact quantum 1-query algorithm with $U_0$, $O_x$ and $U_1$ being the sequence of unitary operators for $f^{(1)}_{2h+1,h}$, and with starting state $|\psi_s\rangle$.
Let $U_0|\psi_s\rangle=\sum_{i=0,j=1}^{n,m}\alpha_{ij}|i\rangle|j\rangle$. When $|x|=0$, we have $|\psi_0\rangle= O_xU_0|\psi_s\rangle=\sum_{i=0,j=1}^{n,m}\alpha_{ij}|i\rangle|j\rangle$.
Denote $\beta_i=\sum_{j=1}^m |\alpha_{ij}|^2$. We prove that $\beta_1=\cdots=\beta_n$ as follows.

Given two inputs $x$ and $y$ such that $|x|=|y|=h$, $x_k\neq y_k$, $x_l\neq y_l$ and $x_i=y_i$ for $i\neq k,l$ (in this case $x_k\neq x_l$ and $y_k\neq y_l$ ), then we have
\begin{align}
0=&(U_1|\psi_0\rangle)^{\dag}U_1O_x|\psi_0\rangle-(U_1|\psi_0\rangle)^{\dag}U_1O_y|\psi_0\rangle\\
=& \langle\psi_0|O_x|\psi_0\rangle-\langle\psi_0|O_y|\psi_0\rangle\\
=& \left(\sum_{i=0,j=1}^{n,m}\overline{\alpha}_{ij}\langle i|\langle j|\right) \\
&\times\left(\sum_{i=1,j=1}^{n,m}(-1)^{x_i}\alpha_{ij}|i\rangle|j\rangle-\sum_{i=1,j=1}^{n,m}(-1)^{y_i}\alpha_{ij}|i\rangle|j\rangle\right)\\
=&(-1)^{x_k}\beta_k+(-1)^{x_l}\beta_l-\left((-1)^{y_k}\beta_k+(-1)^{y_l}\beta_l\right)\\
=&\left((-1)^{x_k}-(-1)^{y_k} \right)(\beta_k-\beta_l)=2(-1)^{x_k}(\beta_k-\beta_l).
\end{align}
Therefore, $\beta_1=\cdots=\beta_n=\beta$ for some $\beta$.

In addition, suppose that the input $x$ satisfies $x_i=1$ for $1\leq i\leq h$, and $x_{j}=0$ for $h+1\leq j\leq n$. Then we have
\begin{align}
0=&(U_1|\psi_0\rangle)^{\dag}U_1O_x|\psi_0\rangle\\
=&\left(\sum_{i=0,j=1}^{n,m}\overline{\alpha}_{ij}\langle i|\langle j|\right) \left(\sum_{j=1}^{m}\alpha_{0j}|0\rangle|j\rangle+\sum_{i=1,j=1}^{n,m}(-1)^{x_i}\alpha_{ij}|i\rangle|j\rangle\right)\\
=&\beta_0+(-1)h\beta+(h+1)\beta=\beta_0+\beta=0,
\end{align}
but this leads to $\beta_0=\cdots=\beta_n=0$ and thus $\alpha_{ij}=0$ for $0\leq i\leq n$ and $1\leq j\leq m$, which is a contradiction. Hence, $Q_E(f^{(1)}_{2h+1,h})\geq 2$.
\end{proof}

Since any exact quantum query algorithm being able to compute $f^{(3)}_{2m+1,m}$ can also compute $f^{(1)}_{2m+1,m}$,
it follows  that $Q_E(f^{(3)}_{2m+1,m})\geq Q_E(f^{(1)}_{2m+1,m})>1$. Moreover, $f^{(3)}_{2m+1,m+1}$ is isomorphic to $f^{(3)}_{2m+1,m}$.
Combining Theorems \ref{T-deg-less2} and \ref{T-floor-n/2} as well as Proposition \ref{GQZ15}, we have the following result.

\begin{theorem}\label{odd-1-query}
Let $n$ be odd and let $f:\{0,1\}^n \to\{0,1\}$  be an $n$-bit symmetrically  partial function. Then $Q_E(f)=1$ if and only if $f$ is isomorphic to one of the functions $f^{(1)}_{n,k}$,
 where $ k\geq \lceil n/2 \rceil$.
\end{theorem}

\begin{proof}
Suppose that $Q_E(f)=1$.  From Lemma \ref{Lemma-lower-bound1} it follows $\text{deg}(f)\leq 2Q_E(f)=2$. Due to Lemma \ref{odd}, $f$ is isomorphic to one of these functions: $f^{(1)}_{n,n}, f^{(1)}_{n,k}, f^{(2)}_{n,k}, f^{(3)}_{n,l}, f^{(4)}_{n,1}$,
where $n-1\geq k\geq \lfloor n/2\rfloor$, and $\lceil n/2\rceil \geq l\geq \lfloor n/2\rfloor$.

By using Proposition \ref{GQZ15},  $Q_E(f^{(1)}_{n,k})=1$ for $n\geq k\geq \lceil n/2\rceil$.

Next we verify the remainder functions have exact query complexity more than 1.

Firstly, Theorem \ref{T-floor-n/2} shows that $Q_E(f^{(1)}_{n,k})>1$ for $k=\lfloor n/2\rfloor$;

Secondly, Theorem \ref{T-k2} verifies that $Q_E(f^{(2)}_{n,k})>1$ for $n-1\geq k\geq \lfloor n/2\rfloor$;

Thirdly, for $l=\lfloor n/2\rfloor$, since any exact quantum query algorithm being able to compute $f^{(3)}_{n,l}$ can also compute $f^{(1)}_{n,l}$,
it follows  that $Q_E(f^{(3)}_{n,l})\geq Q_E(f^{(1)}_{n,l})>1$. Moreover, for $l=\lfloor n/2\rfloor$, $f^{(3)}_{n,l}$ is isomorphic to $f^{(3)}_{n,l+1}$, so, $Q_E(f^{(3)}_{n,l+1})= Q_E(f^{(3)}_{n,l})\geq Q_E(f^{(1)}_{n,l})>1$.

Finally, for $\lceil n/2\rceil \geq l\geq \lfloor n/2\rfloor$, any exact quantum query algorithm computing $f^{(4)}_{n,1}$ can also compute $f^{(3)}_{n,l}$,  so, $Q_E(f^{(4)}_{n,1})\geq Q_E(f^{(3)}_{n,l})>1$.

In a word, $Q_E(f)=1$ implies $f$ is isomorphic to one of the functions $f^{(1)}_{n,k}$, for $ k\geq \lceil n/2 \rceil$.

 On the other hand, since $Q_E(f^{(1)}_{n,k})=1$ for $k\geq \lceil n/2\rceil$,
 if  $f$ is isomorphic to one of these functions $f^{(1)}_{n,k}$ for $ k\geq \lceil n/2 \rceil$,  by {\bf Fact \ref{fact1}}, it holds $Q_E(f)=1$.

\end{proof}

\begin{remark}
 $f^{(3)}_{n,n/2}$ is  the Deutsch-Jozsa problem. Sometimes we can equivalently transform some problems to the Deutsch-Jozsa problem or its more special cases by padding some strings. Indeed, if we pad $2k-n$ zeros to the input of the function $f^{(1)}_{n,k}$, then it is equivalently to solve $f^{(1)}_{2k,k}$ that is simpler and more special than the Deutsch-Jozsa  problem.  Therefore we can use the  Deutsch-Jozsa algorithm to solve the problem. That is to say, a symmetrically  partial Boolean function $f$ has $Q_E(f)=1$ if and only if $f$ can be computed by the Deutsch-Jozsa algorithm after appropriate padding of the input.
\end{remark}

Therefore, with Theorems \ref{even-1-query} and \ref{odd-1-query} we are ready to obtain the second main result of the article:

{\bf Theorem 1.} {\em
Any symmetric and partial Boolean function $f$ has $Q_E(f)=1$ if and only if $f$ can be computed by the Deutsch-Jozsa algorithm.}

Next we further discuss the exact quantum query complexity for the partial symmetric Boolean functions with degree 2.

We have already proved that  $Q_E(f^{(1)}_{2m+1,m})>1$ and $Q_E(f^{(2)}_{2m+1,m})>1$, and furthermore $Q_E(f^{(3)}_{2m+1,m+1})=Q_E(f^{(3)}_{2m+1,m})\geq Q_E(f^{(1)}_{2m+1,m})>1$.

 More investigations concerning  $Q_E(f^{(2)}_{n,k})$  and $Q_E(f^{(4)}_{n})$ will be done in next section. Now we give two optimal algorithms for $f^{(1)}_{2m+1,m}$ and $f^{(3)}_{2m+1,m+1}$ in the following.
%have now that $Q_E(f3_{2m+1}^m)\geq Q_E(f1_{2m+1}^m)>1$ and also that $f3_{2m+1}^{m+1}$ is isomorphic to $f3_{2m+1}^{m}$.

\begin{theorem} \label{QE2}
 $Q_E(f^{(1)}_{2m+1,m})=Q_E(f^{(3)}_{2m+1,l})= 2$, where $m\leq l\leq m+1$.
\end{theorem}
\begin{proof}
Let $n=2m+1$.    Let $\text{DJ}(n, 0, x)$ be the subroutine to compute the Deutsch-Jozsa  problem.  Let $\text{DHW}(n, k, x)$ be the subroutine to compute the function $f^{(1)}_{n,k}$, where $k\geq \lceil n/2\rceil$.
% Let $\text{DHW2}(n, k, x)$ be the subroutine to compute the following function
%\begin{equation}
%f5_n^k=\left\{\begin{array}{ll}
%                    0 &\text{if}\ |x|=n \\
%                    1 &\text{if}\ |x|= k,
%                  \end{array}
% \right.
%\end{equation}
%where $k\leq \lfloor n/2\rfloor$. It is clear that $f5_n^k$ is isomorphic to $f1_n^{n-k}$, where $k\leq \lfloor n/2\rfloor$. Therefore $Q_E(f5_n^k)=1$ and we can use a similar algorithm to compute $f5_n^k$.
  As we knew, the exact quantum algorithms just use 1 query in the above  subroutines.
We now  give an exact quantum 2-query algorithm to compute $f^{(1)}_{2m+1,m}$ as Algorithm \ref{Algorithm-F1}.
\begin{algorithm}
\caption{Algorithm for $f^{(1)}_{2m+1,m}$}
\label{Algorithm-F1}
\begin{algorithmic}[1]
 \Procedure{f1} {{\bf integer} $n$, {\bf array} $x$}\Comment{$x\in\{0,1\}^n$}
\State Query $x_1$
\If{$x_1=1$} \Return 1
\EndIf
\If{$x_1=0$}
\State  $x\gets x\setminus\{x_1\}$
\State \Return $\text{DHW}(n-1, \lfloor n/2 \rfloor, x)$
\EndIf
\EndProcedure
\end{algorithmic}
\end{algorithm}

We give an   exact quantum  2-query algorithm to compute $f^{(3)}_{2m+1,m+1}$  as Algorithm \ref{Algorithm-F3}.
\begin{algorithm}
\caption{Algorithm for $f^{(3)}_{2m+1,m+1}$}
\label{Algorithm-F3}
\begin{algorithmic}[1]
 \Procedure{f3} {{\bf integer} $n$,  {\bf array} $x$}\Comment{$x\in\{0,1\}^n$}
\State Query $x_1$
\State  $x\gets x\setminus\{x_1\}$
\If{$x_1=1$}
\State \Return $\text{DJ}(n-1, 0, x)$
\EndIf
\If{$x_1=0$}
\State \Return $\text{DHW}(n-1, \lceil n/2\rceil, x)$
\EndIf
\EndProcedure
\end{algorithmic}
\end{algorithm}

It is clear that the above two algorithms can compute the functions $f^{(1)}_{2m+1,m}$ and $f^{(3)}_{2m+1,m}$ with two queries, respectively.
The function $f^{(3)}_{2m+1,m}$ is isomorphic to the function $f^{(3)}_{2m+1,m+1}$, so it holds that $Q_E(f^{(3)}_{2m+1,m})= 2$.
\end{proof}

\begin{remark}
From Theorem \ref{QE2} it follows the optimal exact quantum query complexity is 2 for solving a variant of the  Deutsch-Jozsa problem, i.e. distinguishes between inputs of Hamming weight in $\{ \lfloor n/2\rfloor, \lceil n/2\rceil \}$ and Hamming weight in $\{ 0, n\}$ for  all odd $n$.
\end{remark}

\section{Further results}

This section first studies  $n$-bit symmetrically  partial  function $\text{DW}_n^{k,l}$:
\begin{equation}
\text{DW}_n^{k,l}(x)=\left\{\begin{array}{ll}
                    0 &\text{if}\ |x|=k, \\
                    1 &\text{if}\ |x|=l.
                  \end{array}
 \right.
\end{equation}
We will give some optimal exact quantum query algorithms to  compute the function $\text{DW}_n^{k,l}$ for some special choices of $k$ and $l$, and then use $\text{DW}_n^{k,l}$ as subroutines to give some exact quantum query algorithms to compute the functions   $f^{(2)}_{n,k}$ and $f^{(4)}_{n}$.

\begin{theorem}\label{Th-4m-m-3m}
 $Q_E(\text{{\em DW}}_{4m}^{m,3m})=2$.
\end{theorem}
\begin{proof}

We give a 2-query exact quantum algorithm to compute $\text{DW}_{4m}^{m,3m}$ as Algorithm \ref{A-DW-4m}, where the subroutine Grover($n,x$) is a 1-query Grover search \cite{Gro96} which returns an index $i$. We describe Grover($n,x$) as follows:
\begin{enumerate}
\item [(1)] Begin with the quantum state $|1\rangle$, and a unitary transformation acts on it, resulting in $|\psi_0\rangle=W|1\rangle=\frac{1}{\sqrt{n}}\sum_{i=1}^n|i\rangle$.
\item [(2)] Act on the quantum state $|\psi_0\rangle$ with the transformation $G=-WZ_1W^{\dag}Z_f$, where
\begin{equation}
Z_1|i\rangle=\left\{\begin{array}{ll}
                    -|i\rangle \ \text{if } i=1, \\
                   \ \ \ |i\rangle \  \text{if } i\neq 1, \\
                  \end{array}
 \right. \ \text{and\  \ }
Z_f|i\rangle=(-1)^{x_i}|i\rangle.
\end{equation}
\item [(3)] Measure the quantum state with the projective measurement $\{|i\rangle\langle i|\}_{i=1}^{n}$, then returning the measurement result $i$.
\end{enumerate}
Let us consider the quantum state after the transformation $G$:
\begin{align}
|\psi\rangle &=G|\psi_0\rangle=-WZ_1W^{\dag}Z_f\frac{1}{\sqrt{n}}\sum_{i=1}^n|i\rangle \\
 &=-W(I-2|1\rangle\langle1|)W^{\dag}Z_f\left(\frac{1}{\sqrt{n}}\sum_{i=1}^n|i\rangle\right)\\
&=(2W|1\rangle\langle1|W^{\dag}-I) \left(\frac{1}{\sqrt{n}}\sum_{i=1}^n(-1)^{x_i}|i\rangle\right)\\
&=\left(2\left(\frac{1}{\sqrt{n}}\sum_{i=1}^n|i\rangle\right) \left(\frac{1}{\sqrt{n}}\sum_{i=1}^n\langle i|\right) -I\right) \left(\frac{1}{\sqrt{n}}\sum_{i=1}^n(-1)^{x_i}|i\rangle\right)\\
&=\frac{1}{\sqrt{n}}\sum_{i=1}^n\left( \left(\frac{2}{n}\sum_{j=1}^n(-1)^{x_j}\right)-(-1)^{x_i}\right)|i\rangle.
\end{align}
When $|x|=n/4$, we have $\frac{2}{n}\sum_{j=1}^n(-1)^{x_j}=1$ and $|\psi\rangle=\frac{2}{\sqrt{n}}\sum_{i:x_i=1}|i\rangle$. When $|x|=3n/4$, we have $\frac{2}{n}\sum_{j=1}^n(-1)^{x_j}=-1$ and $|\psi\rangle=\frac{-2}{\sqrt{n}}\sum_{i:x_i=0}|i\rangle$. Therefore, after the measurement, the subroutine Grover($n,x$) will return an index $i$ such that $x_i=1$ if $|x|=n/4$ and $x_i=0$ if
$|x|=3n/4$.

\begin{algorithm}
\caption{Algorithm for $\text{DW}_{4m}^{m,3m}$}
\label{A-DW-4m}
\begin{algorithmic}[1]
 \Procedure{$\text{DW1}$} {{\bf integer} $n$,  {\bf array} $x$}\Comment{$x\in\{0,1\}^n$}
\State $i\gets$ Grover($n,x$)
\State Query $x_i$
\State \Return $1-x_i$.
\EndProcedure
\end{algorithmic}
\end{algorithm}

According to the above analysis, it is clear that Algorithm \ref{A-DW-4m} computes the function $\text{DW}_{4m}^{m,3m}$ with 2 queries and thus $Q_E(\text{DW}_{4m}^{m,3m})\leq 2$.  We now prove $Q_E(\text{DW}_{4m}^{m,3m})\geq 2$ as follows. We note that the function $\text{DW}_{4m}^{m,3m}$ is not isomorphic to any function  pointed out in Theorem \ref{T-deg-less2}. Therefore, we have $\text{deg}(\text{DW}_{4m}^{m,3m})>2$ and $Q_E(\text{DW}_{4m}^{m,3m})> 1$ as well. Therefore, $Q_E(\text{DW}_{4m}^{m,3m})= 2$.
\end{proof}

Using the padding method (similar method used in \cite{BBD+15}), we have
\begin{corollary}
For any integers $n$ and $k$ such that $0<k< n/3$,  $l\geq \max\{(2n+k)/3,3k\}$ and $l-k$ is even,  then $Q_E(\text{{\em DW}}_n^{k,l})=2$.
\end{corollary}
\begin{proof}
By padding $(3l-k)/2-n$ zeroes and $(l-3k)/2$ ones to the inputs of the function $\text{DW}_n^{k,l}$, we have the new $n'=n+(3l-k)/2-n+(l-3k)/2=2(l-k)$.
The new $k'=k+(l-3k)/2=(l-k)/2$ and the new $l'=l+(l-3k)/2=3(l-k)/2$.   Therefore, the new function is $\text{DW}_{n'}^{n'/4,3n'/4}$ and, from Theorem \ref{Th-4m-m-3m} it follows that $Q_E(\text{DW}_n^{k,l})=2$.
\end{proof}

Let $k\leq n/4$.  Then $\max\{(2n+k)/3,3k\}=3n/4$. We have the following result.

\begin{corollary}
For any integers $n$ and $k$ such that $0<k\leq n/4$, $3n/4\leq l<n$ and $l-k$ is even,   then $Q_E(\text{{\em DW}}_n^{k,l})=2$.
\end{corollary}

The next theorem is
implicit in the combination of the proof in \cite{BBD+15} and  Theorem \ref{T-deg-less2}.

\begin{theorem}
 For any integers $n$ and $l$ such that $\frac{n}{4}\leq l<\lfloor \frac{n}{2}\rfloor $, then $Q_E(\text{{\em DW}}_{n}^{0,l})=2$.
\end{theorem}
\begin{proof}
Padding $(4l-n)$ zeros to the input, we have $n'=n+4l-n=4l$.  The function changes to $\text{DW}_{4l}^{0,l}$. We give an  exact quantum 2-query algorithm to compute $\text{DW}_{4l}^{0,l}$ as Algorithm \ref{A-DW-0-m}.
\begin{algorithm}
\caption{Algorithm for $\text{DW}_{4l}^{0,l}$}
\label{A-DW-0-m}
\begin{algorithmic}[1]
 \Procedure{$\text{DW2}$} {{\bf integer} $n$,  {\bf array} $x$}\Comment{$x\in\{0,1\}^n$}
\State $i\gets$ Grover($n,x$)
\State Query $x_i$
\State \Return $x_i$.
\EndProcedure
\end{algorithmic}
\end{algorithm}
Similar to the analysis in Theorem \ref{Th-4m-m-3m}, it is clear that Algorithm \ref{A-DW-0-m} computes $\text{DW}_{4l}^{0,l}$. Therefore, we have $Q_E(\text{DW}_{n}^{0,l})\leq 2$ for  $\frac{n}{4}\leq l<\lfloor \frac{n}{2}\rfloor $. According to Theorem \ref{T-deg-less2}, we have $\text{deg}(\text{DW}_{n}^{0,l})>2$ for $\frac{n}{4}\leq l<\lfloor \frac{n}{2}\rfloor $. Therefore, we have  $Q_E(\text{DW}_{n}^{0,l})\geq 2$ for $\frac{n}{4}\leq l<\lfloor \frac{n}{2}\rfloor $.
\end{proof}

Now we use $\text{DW}_n^{k,l}$ as subroutines to give some algorithms for the functions that we discussed in Section \ref{section-degree}.
\begin{theorem}\label{T-k2-less4}
For $n/4\leq k<n$, $Q_E(f^{(2)}_{n,k})\leq 4$, where
\begin{equation}
    f^{(2)}_{n,k}(x)=\left\{\begin{array}{ll}
                    0 &\text{if}\ |x|=0, \\
                    1 &\text{if}\ |x|= k\ \text{\em or } |x|=k+1.
                  \end{array}
 \right.
\end{equation}
\end{theorem}

\begin{proof}
We give an  exact quantum 4-query algorithm for the function as Algorithm \ref{A-f2}.
\begin{algorithm}
\caption{Algorithm for $f^{(2)}_{n,k}$}
\label{A-f2}
\begin{algorithmic}[1]
 \Procedure{f2} {{\bf integer} $n$, {\bf integer} $k$,  {\bf array} $x$}\Comment{$x\in\{0,1\}^n$}
 \State Pad $4k-n$ zeros to the input $x$ and get a new input $y$.
\State $i\gets$ Grover($4k,y$)
\State Query $x_i$
\If{$x_i=1$}
\State \Return 1
\EndIf
\State Pad $4(k+1)-n$ zeros to the input $x$ and get a new input $y$.
\State $i\gets$ Grover($4(k+1),y$)
\State Query $x_i$
\State \Return $x_i$
\EndProcedure
\end{algorithmic}
\end{algorithm}

It is clear that  Algorithm \ref{A-f2} can computes the function $f^{(2)}_{n,k}$.
\end{proof}

\begin{theorem}\label{T-f4-less5}
 For any odd $n$, $Q_E(f^{(4)}_{n})\leq 5$, where
\begin{equation}
   f^{(4)}_{n}(x)=\left\{\begin{array}{ll}
                    0 &\text{if}\ |x|=0 \ \text{\em or } |x|=n,\\
                    1 &\text{if}\ |x|=\lfloor n/2\rfloor \ \text{\em or } |x|=\lceil n/2\rceil.
                  \end{array}
 \right.
\end{equation}
\end{theorem}

\begin{proof}
We give a 5-query exact quantum algorithm for the function as Algorithm \ref{A-f4}.
\begin{algorithm}
\caption{Algorithm for $f^{(4)}_{n}$}
\label{A-f4}
\begin{algorithmic}[1]
 \Procedure{f4} {{\bf integer} $n$,  {\bf array} $x$}\Comment{$x\in\{0,1\}^n$}
\State Query $x_1$
\State  $x\gets x\setminus\{x_1\}$
\If{$x_1=0$}
\State \Return $\text{F2}(n-1, \lfloor n/2\rfloor, x)$
\EndIf
\If{$x_1=1$}
\State Let the new input $y=\overline{x}$
\State \Return $\text{F2}(n-1, \lfloor n/2\rfloor, y)$
\EndIf
\EndProcedure
\end{algorithmic}
\end{algorithm}

If $x_1=1$, then the function reduces to the following function:
\begin{equation}
   f^{(4)}_{n-1,1}(x)=\left\{\begin{array}{ll}
                    0 &\text{if}\ |x|=n-1,\\
                    1 &\text{if}\ |x|=\lfloor n/2\rfloor-1 \ \text{\em or } |x|=\lceil n/2\rceil-1,
                  \end{array}
 \right.
\end{equation}
which is isomorphic to $f_{n-1,\lfloor n/2\rfloor}^{(2)}$.
It is clear that  Algorithm \ref{A-f4} can compute the function $f^{(4)}_{n}$.
\end{proof}

\section*{Acknowledgements}
We thank Lvzhou Li for much essential discussions on  query complexity and  this article.
This work  was  supported by the National
Natural Science Foundation of China (Nos.   61572532, 61272058, 61602532) and the Fundamental Research Funds for the Central Universities of China (Nos. 17lgjc24, 161gpy43).

\section*{Appendix A: The subroutine for Xquery }\label{Xquery}

The subroutine will use basis state $|0,0\rangle$, $|i,0\rangle$ and  $|i,j\rangle$ with $1\leq i<j\leq m$.
\begin{enumerate}
  \item
The subroutine $\mbox{Xquery}$  begins in the  state $|0,0\rangle$ and then a unitary mapping $U_1$ is applied on it:
\begin{equation}
    U_1|0,0\rangle=\sum_{i=1}^m\frac{1}{\sqrt{m}}|i,0\rangle.
\end{equation}
 \item
The subroutine $\mbox{Xquery}$  then  performs the query:
\begin{equation}
 \sum_{i=1}^m\frac{1}{\sqrt{m}}|i,0\rangle\to  \sum_{i=1}^m\frac{1}{\sqrt{m}}(-1)^{x_i}|i,0\rangle.
\end{equation}

 \item The subroutine $\mbox{Xquery}$  performs a unitary mapping $U_2$ to the current state such that
 \begin{equation}
 U_2|i,0\rangle=\sum_{j>i}\frac{1}{\sqrt{m}}|i,j\rangle-\sum_{ j<i}\frac{1}{\sqrt{m}}|j,i\rangle+\frac{1}{\sqrt{m}}|0,0\rangle
 \end{equation}
 and the resulting quantum state will be
  \begin{align}
  &U_2\sum_{i=1}^m\frac{1}{\sqrt{m}}(-1)^{x_i}|i,0\rangle\nonumber\\
  =&\frac{1}{m}\sum_{i=1}^m(-1)^{x_i}|0,0\rangle+\frac{1}{m}\sum_{1\leq i<j}((-1)^{x_i}-(-1)^{x_j})|i,j\rangle.
   \end{align}
   \item The subroutine $\mbox{Xquery}$   measures the resulting state in the standard basis.  If the outcome is  $|0,0\rangle$, then $\sum_{i=1}^m(-1)^{x_i}\neq 0 $ and $|x|\neq m/2$. Otherwise, suppose that we get the state $|i,j\rangle$. Then we have $x_i\neq x_j$ and the subroutine outputs $(i,j)$.
\end{enumerate}

\section*{Appendix B: Proof of Equality (\ref{eq-det})}\label{details}

\begin{proof}
We define ${p \choose l}=0$ if $p<l$ and also ${p \choose l}=0$ if $l<0$.  For any integers $p$ and $l$,  it is easy to see that ${p \choose l}={p \choose p-l}$.   Now we prove that for any integers $p$ and $l$,
\begin{equation}
   (p+1){p \choose l}=(l+1){p+1 \choose l+1}.
\end{equation}
There are several cases as follows:
\begin{enumerate}
\item [Case 1]  $p<l$.  In this case, ${p \choose l}=0$ and ${p+1 \choose l+1}$, the equality holds.
  \item [Case 2] $p\geq l\geq 0$. In this case, $(p+1){p \choose l}=(p+1)\frac{p!}{l!(p-l)!}=(l+1)\frac{(p+1)!}{(l+1)!((p+1)-(l+1))!}=(l+1){p+1 \choose l+1}$.
  \item [Case 3] $p\geq l$ and $l<-1$. In this case,  ${p+1 \choose l+1}=0$ and ${p \choose l}=0$, the equality holds.
  \item [Case 4] $p\geq l$ and $l=-1$. In this case,   ${p \choose l}=0$ and $(l+1){p+1 \choose l+1}=0$, the equality holds.
\end{enumerate}
Therefore, the  equality holds.

Now we are ready to prove Equality (\ref{eq-det}).
\begin{align}
   &\left|
      \begin{array}{ccccc}
        {n \choose k+1} & {n \choose k+2}  & \cdots &  {n \choose 2k}   &  {n \choose 2k+1}\\
       {n-1 \choose k+1}& {n-1 \choose k+2}& \cdots & {n-1 \choose 2k}& {n-1 \choose 2k+1}\\
       \vdots& \vdots& \ddots & \vdots& \vdots\\
       {n-k+1 \choose k+1} & {n-k+1 \choose k+2} & \cdots &{n-k+1 \choose 2k}  &{n-k+1 \choose 2k+1}\\
       {n-k \choose k+1} & {n-k \choose k+2} & \cdots  &{n-k \choose 2k} &{n-k \choose 2k+1}\\
      \end{array}
    \right|\nonumber\\
&=\left|
      \begin{array}{ccccc}
        {n \choose n-k-1} & {n \choose n-k-2}  & \cdots &  {n \choose n-2k} &  {n \choose n-2k-1}\\
       {n-1 \choose n-k-2}& {n-1 \choose n-k-3}& \cdots & {n-1 \choose n-2k-1} & {n-1 \choose n-2k-2}\\
       \vdots& \vdots& \ddots & \vdots& \vdots\\
       {n-k+1 \choose n-2k} & {n-k+1 \choose n-2k-1} & \cdots &{n-k+1 \choose n-3k-1} &{n-k+1 \choose n-3k}\\
       {n-k \choose n-2k-1} & {n-k \choose n-2k-2} & \cdots  &{n-k \choose n-3k}  &{n-k \choose n-3k-1}\\
      \end{array}
    \right|\nonumber\\
&=\frac{1}{(n-k+1)}\times \nonumber\\
&\left|
      \begin{array}{ccccc}
        {n \choose n-k-1} & {n \choose n-k-2}  & \cdots &  {n \choose n-2k} &   {n \choose n-2k-1}\\
       {n-1 \choose n-k-2}& {n-1 \choose n-k-3}& \cdots  & {n-1 \choose n-2k-1}& {n-1 \choose n-2k-2}\\
       \vdots& \vdots& \ddots & \vdots&\vdots\\
       {n-k+1 \choose n-2k} & {n-k+1 \choose n-2k-1} & \cdots  &{n-k+1 \choose n-3k-1} &{n-k+1 \choose n-3k}\\
       (n-k+1){n-k \choose n-2k-1} & (n-k+1){n-k \choose n-2k-2} & \cdots &(n-k+1){n-k \choose n-3k} &(n-k+1){n-k \choose n-3k-1}
      \end{array}
    \right|\nonumber\\
&=\frac{1}{(n-k+1)}\times \nonumber\\
&\left|
      \begin{array}{ccccc}
        {n \choose n-k-1} & {n \choose n-k-2}  & \cdots  &  {n \choose n-2k}&  {n \choose n-2k-1}\\
       {n-1 \choose n-k-2}& {n-1 \choose n-k-3}& \cdots  & {n-1 \choose n-2k-1}& {n-1 \choose n-2k-2}\\
       \vdots& \vdots& \ddots & \vdots& \vdots\\
       {n-k+1 \choose n-2k} & {n-k+1 \choose n-2k-1} & \cdots &{n-k+1 \choose n-3k-1} &{n-k+1 \choose n-3k}\\
       (n-2k-1){n-k+1 \choose n-2k} & (n-2k-2){n-k+1 \choose n-2k-1} & \cdots &(n-3k){n-k+1 \choose n-3k-1} &(n-3k-1){n-k+1 \choose n-3k}\\
      \end{array}
    \right|\nonumber\\
    \end{align}
\begin{align}
&\xlongequal{r_{k+1}-(n-3k-1)r_{k}}\frac{1}{(n-k+1)}\left|
      \begin{array}{ccccc}
        {n \choose n-k-1} & {n \choose n-k-2}  & \cdots &  {n \choose n-2k} &  {n \choose n-2k-1}\\
       {n-1 \choose n-k-2}& {n-1 \choose n-k-3}& \cdots & {n-1 \choose n-2k-1} & {n-1 \choose n-2k-2}\\
       \vdots& \vdots& \ddots & \vdots & \vdots \\
       {n-k+1 \choose n-2k} & {n-k+1 \choose n-2k-1} & \cdots &{n-k+1 \choose n-3k-1}  &{n-k+1 \choose n-3k}\\
       k{n-k+1 \choose n-2k} & (k-1){n-k+1 \choose n-2k-1}  & \cdots &{n-k+1 \choose n-3k-1}  &0\\
      \end{array}
    \right|\nonumber\\
&\xlongequal{\cdots}\frac{1}{(n-k+1)(n-k+2)\cdots n}\left|
      \begin{array}{ccccc}
        {n \choose n-k-1} & {n \choose n-k-2}  & \cdots &  {n \choose n-2k}  &  {n \choose n-2k-1}\\
       k{n \choose n-k-1}& (k-1){n \choose n-k-2}& \cdots & {n \choose n-2k} &0\\
       \vdots& \vdots& \ddots & \vdots& \vdots\\
       k{n-k+2 \choose n-2k+1} & (k-1){n-k+2 \choose n-2k} & \cdots & {n-k+2 \choose n-3k} &0\\
       k{n-k+1 \choose n-2k} & (k-1){n-k+1 \choose n-2k-1} & \cdots &{n-k+1 \choose n-3k-1}&0\\
      \end{array}
    \right|\nonumber
    \end{align}
    \begin{align}
&=\frac{k(k-1)\cdots 1}{(n-k+1)(n-k+2)\cdots n}\left|
      \begin{array}{ccccc}
        \frac{1}{k}{n \choose n-k-1} &  \frac{1}{k-1}{n \choose n-k-2}  & \cdots &   \frac{1}{1}{n \choose n-2k}  &  {n \choose n-2k-1}\\
       {n \choose n-k-1}& {n \choose n-k-2}& \cdots & {n \choose n-2k} &0\\
       \vdots& \vdots& \ddots & \vdots& \vdots\\
       {n-k+2 \choose n-2k+1} & {n-k+2 \choose n-2k} & \cdots & {n-k+2 \choose n-3k} &0\\
       {n-k+1 \choose n-2k} & {n-k+1 \choose n-2k-1} & \cdots &{n-k+1 \choose n-3k-1}&0\\
      \end{array}
    \right| \nonumber\\
&=\frac{1}{{n \choose k}}
\left|
      \begin{array}{ccccc}
        \frac{1}{k}{n \choose n-k-1} &  \frac{1}{k-1}{n \choose n-k-2}  & \cdots &   \frac{1}{1}{n \choose n-2k}  &  {n \choose n-2k-1}\\
       {n \choose n-k-1}& {n \choose n-k-2}& \cdots & {n \choose n-2k} &0\\
       \vdots& \vdots& \ddots & \vdots& \vdots\\
       {n-k+2 \choose n-2k+1} & {n-k+2 \choose n-2k} & \cdots & {n-k+2 \choose n-3k} &0\\
       {n-k+1 \choose n-2k} & {n-k+1 \choose n-2k-1} & \cdots &{n-k+1 \choose n-3k-1}&0\\
      \end{array}
    \right|\nonumber \\
&=(-1)^{k+2}\cdot\frac{{n \choose n-2k-1}}{{n \choose k}}
\left|
      \begin{array}{cccc}
       {n \choose n-k-1}& {n \choose n-k-2}& \cdots & {n \choose n-2k} \\
       \vdots& \vdots& \ddots & \vdots\\
       {n-k+2 \choose n-2k+1} & {n-k+2 \choose n-2k} & \cdots & {n-k+2 \choose n-3k} \\
       {n-k+1 \choose n-2k} & {n-k+1 \choose n-2k-1} & \cdots &{n-k+1 \choose n-3k-1}\\
      \end{array}
    \right| \nonumber\\
&\xlongequal{\cdots}
(-1)^{k+2}(-1)^{k+1}\cdots (-1)^{3}
\frac{{n \choose n-2k-1}{n \choose n-2k}\cdots {n \choose n-k-2} }{{n \choose k}{n \choose k-1}\cdots {n \choose 1}}\left|
     {n \choose n-k-1}
    \right|\nonumber \\
&=(-1)^{\frac{k(k+5)}{2}}\cdot\frac{\prod_{i=k+1}^{2k+1}{n \choose i}}{\prod_{i=1}^{k}{n \choose i}}\neq 0.\nonumber
\end{align}

\end{proof}

\end{document}